\documentclass{lmcs}
\pdfoutput=1

\usepackage{lastpage}
\lmcsdoi{20}{3}{3}
\lmcsheading{}{\pageref{LastPage}}{}{}%
{May~17,~2022}{Jul.~03,~2024}{}


\usepackage{hyperref}
\usepackage[utf8]{inputenc}
\usepackage{xspace}

\usepackage{stmaryrd,bclogo,amsthm}
\usepackage{mathpartir}

\usepackage{mathabx}
\usepackage[retainmissing]{MnSymbol}

\usepackage{tikz}
\usepackage{tikz-qtree}

\usepackage{macros}

\begin{document}

\title[Encodings of the $\lambda$-Calculus in \hocore]{Fully Abstract Encodings of
  \texorpdfstring{$\lambda$-Calculus \\}{} in \hocore through
  Abstract Machines
} 

\author[M.~Biernacka]{Ma\l{}gorzata Biernacka \lmcsorcid{0000-0001-8094-0980}}[a]	
\author[D.~Biernacki]{Dariusz Biernacki \lmcsorcid{0000-0002-1477-4635}}[a]	
\author[S.~Lenglet]{Sergue\"i Lenglet \lmcsorcid{0000-0001-5588-1050}}[b]	
\author[P.~Polesiuk]{Piotr Polesiuk \lmcsorcid{0000-0002-7012-4346}}[a]
\author[D.~Pous]{Damien Pous \lmcsorcid{0000-0002-1220-4399}}[c]
\author[A.~Schmitt]{Alan Schmitt \lmcsorcid{0000-0002-2551-8039}}[d]

\address{University of Wroc\l{}aw}	
\email{mabi@cs.uni.wroc.pl, dabi@cs.uni.wroc.pl, ppolesiuk@cs.uni.wroc.pl}
\address{Universit\'e de Lorraine}
\address{Plume, LIP, CNRS, ENS de Lyon, Université de Lyon}
\email{Damien.Pous@ens-lyon.fr}	
\address{INRIA}
\email{alan.schmitt@inria.fr}

\begin{abstract}
  We present fully abstract encodings of the call-by-name and
  call-by-value $\lambda$-calculus into \hocore, a minimal
  higher-order process calculus with no name restriction. We consider
  several equivalences on the $\lambda$-calculus side---normal-form
  bisimilarity, applicative bisimilarity, and contextual
  equivalence---that we internalize into abstract machines in order to
  prove full abstraction of the encodings. We also demonstrate that
  this technique scales to the $\lambda\mu$-calculus, i.e., a standard
  extension of the $\lambda$-calculus with control operators.
\end{abstract}

\maketitle

\section{Introduction}
\label{sec:intro}

\hocore is a minimal process calculus with higher-order communication,
meaning that messages are executable processes. It is a subcalculus of
HO$\pi$~\cite{Sangiorgi:IaC96} with no construct to generate names or
to restrict the scope of communication channels. Even with such a
limited syntax, \hocore{} is Turing
complete~\cite{Lanese-al:IC11}. However, as a higher-order calculus,
it is less expressive than the name passing $\pi$-calculus: polyadic
message sending cannot be compositionally encoded in monadic message
sending in HO$\pi$~\cite{Lanese-al:ICALP10}, while it can be done in
$\pi$~\cite{Sangiorgi-Walker:01}.

Although \hocore is Turing complete, a fully abstract encoding of the
$\lambda$-calculus into \hocore{} appears impossible at first.
Indeed, a $\lambda$-term potentially has an unbounded number of
redexes. A straightforward encoding would use communication to emulate
$\beta$-reduction, but since \hocore{} does not provide means to
restrict the scope of communication, one would need as many distinct
names as there are redexes to avoid interference. Moreover, as new
redexes may be created by $\beta$-reduction, we also need a way to
generate new names on which to communicate. To circumvent these
problems and illustrate the expressiveness of \hocore, we consider
encodings where the \emph{reduction strategy} is fixed, thus for which
at most one redex is enabled at any time.  In this setting,
$\beta$-reduction can be emulated using communication on a single,
shared, name. A first contribution of this paper is the
definition of novel encodings of the call-by-name and call-by-value
$\lambda$-calculus---more precisely the Krivine Abstract Machine
(KAM)~\cite{Krivine:HOSC07} and CK
machine~\cite{Felleisen-al:09}---into \hocore.

A faithful encoding not only reflects the operational semantics of a
calculus, it should also reflect its equivalences. Ideally, an
encoding is \emph{fully abstract}: two source terms are behaviorally
equivalent iff their translations are equivalent. On the \hocore side,
we use \emph{barbed equivalence with hidden
  names}~\cite{Lanese-al:IC11}, where a fixed number of names used for
the translation cannot be observed. On the $\lambda$-calculus side, we
consider three equivalences. First, we look at \emph{normal-form
  bisimilarity}~\cite{Lassen:MFPS99}, where normal forms
are decomposed into subterms that must be bisimilar. Next, we turn to
\emph{applicative bisimilarity}~\cite{Abramsky-Ong:IaC93}, where
normal forms must behave similarly when applied to identical
arguments. And finally, we consider contextual equivalence, where
terms must behave identically when put into arbitrary contexts.

Our second contribution is an \emph{internalization} of these
equivalences into extended abstract machines: these machines expand
the underlying machine evaluating the term (the KAM or CK machine)
with additional transitions with flags interpreting the
equivalences. By doing so, we can express these different equivalences
on terms by a simpler bisimilarity on the flag-generating machines,
which can be seen as labeled transition systems (LTS). We then
translate these extended machines into \hocore and prove full
abstraction for all three equivalences. The result of the
translation is therefore a \hocore process representing the source
$\lambda$-term augmented with the machine of the internalized
equivalence under consideration. A surrounding context may then
interact with that process by communicating on the machine flags.

Finally, we show that the internalization principle scales to the
$\lambda\mu$-calculus~\cite{Parigot:LPAR92}, an extension of the
$\lambda$-calculus with control operators. The $\mu$ operator is able
to capture the current continuation for later use. Its semantics can
be defined with an extended KAM, in which we can internalize
normal-form and applicative bisimilarities as well as contextual
equivalence. Altogether, this work shows that a minimal process
calculus with no name restriction can faithfully encode the
call-by-name and call-by-value $\lambda$-calculus, with or without
control operators. 

\paragraph*{The chosen equivalences} It is enough to faithfully encode
contextual equivalence to get full abstraction. We choose to study
other equivalences (normal-form and applicative bisimilarities) for
two reasons. First, we start with normal-form bisimilarity because it
is the simplest to translate, as we do not need to inject terms from
the environment to establish the equivalence. We next show how we can
inject terms for applicative bisimilarity, and we finally extend this
approach to contexts for contextual equivalence. Second, the study of
quite different equivalences illustrate the robustness of the
internalization technique. 

\paragraph*{Related work} Since Milner's seminal
work~\cite{Milner:MSCS92}, other encodings of $\lambda$ into $\pi$
have been proposed either as a result in itself~\cite{Cai-Fu:MSCS11,Durier-al:TCS22},
or to study other properties such as connections with
logic~\cite{Accattoli:TERMGRAPH13,Beffara:CoRR11,Toninho-al:FOSSACS12},
termination~\cite{Cimini-al:TGC10,Amadio:CoRR11,Yoshida-al:IaC04},
sequentiality~\cite{Berger-al:TLCA01},
control~\cite{Cimini-al:TGC10,Honda-al:RTATLCA14,vanBakel-Vigliotti:CLC14},
references~\cite{Prebet:ICALP22},
or Continuation-Passing Style (CPS)
transforms~\cite{Sangiorgi:MSCS99,Sangiorgi-Walker:01,Downen-al:PPDP14}. These
works use the more expressive \emph{first-order} $\pi$-calculus,
except for~\cite{Sangiorgi:MSCS99,Sangiorgi-Walker:01}, discussed
below; full abstraction is proved w.r.t. contextual equivalence
in~\cite{Berger-al:TLCA01,Yoshida-al:IaC04,Honda-al:RTATLCA14},
normal-form bisimilarity in~\cite{vanBakel-Vigliotti:CLC14,Durier-al:TCS22,Prebet:ICALP22},
applicative bisimilarity in~\cite{Cai-Fu:MSCS11}, and both
bisimilarities in~\cite{Sangiorgi-Walker:01}. The encodings
of~\cite{Berger-al:TLCA01,Yoshida-al:IaC04,Honda-al:RTATLCA14} are
driven by types, and therefore cannot be compared to our untyped
setting.  In~\cite{vanBakel-Vigliotti:CLC14}, van Bakel et
al. establish a full abstraction result between the
$\lambda\mu$-calculus with normal-form bisimilarity and the
$\pi$-calculus. Their encoding relies on an unbounded number of
restricted names to evaluate several translations of $\lambda$-terms
in parallel, while we rely on flags and on barbed equivalence. We
explain the differences between the two approaches in
Section~\ref{ssec:am-cbn}. The encoding of~\cite{Cai-Fu:MSCS11} also
uses an unbounded number of restricted names, to represent a
$\lambda$-term as a tree and to process it.

Sangiorgi translates the $\lambda$-calculus into a higher-order
calculus as an intermediary step
in~\cite{Sangiorgi:MSCS99,Sangiorgi-Walker:01}, but it is an
abstraction-passing calculus, which is strictly more expressive than a
process-passing calculus~\cite{Lanese-al:ICALP10}. Like in our work,
Sangiorgi fixes the evaluation strategy in the $\lambda$-calculus,
except that he uses CPS translations rather than abstract machines. In
the light of Danvy et al.'s functional
correspondence~\cite{Ager-al:PPDP03}, the two approaches appear
closely related, however it is difficult to compare our encoding with
Sangiorgi's, since we target different calculi, and we internalize the
bisimilarities in the abstract machines. Still, name restriction plays
an important role in Sangiorgi's encodings, since a local channel is
used for each application in a $\lambda$-term. The encoding is fully
abstract w.r.t. normal-form bisimilarity~\cite[Chapter
  18]{Sangiorgi-Walker:01} but not w.r.t. applicative
bisimilarity~\cite[Chapter 17]{Sangiorgi-Walker:01}. Indeed, a
translated $\lambda$-abstraction waits for the outside to provide an
access to an encoded argument to be applied to. However, the
environment may give access to a random process and not to a
translated $\lambda$-term. The encoding of~\cite{Sangiorgi-Walker:01}
does not protect itself against this unwanted behavior from the
environment. In contrast, the encoding of
Section~\ref{sec:full-abstr-applicative} and the one
in~\cite{Cai-Fu:MSCS11} are fully abstract w.r.t. applicative
bisimilarity, because they take this issue into account, as we explain
in Section~\ref{ssec:app-machine}.

\paragraph*{Outline}
Section~\ref{sec:hocore} describes \hocore. We present the main ideas
behind our encodings in Section~\ref{sec:KAM} on the KAM. We explain
how to internalize equivalences into abstract machines and how to
translate the resulting machines for three languages: call-by-name
$\lambda$-calculus (Sections~\ref{sec:full-abs}
and~\ref{sec:full-abstr-applicative}), call-by-value
$\lambda$-calculus (Section~\ref{sec:cbv}), and call-by-name
$\lambda\mu$-calculus (Section~\ref{sec:control}). For each of them,
we discuss normal-form and applicative bisimilarities, and 
contextual equivalence. Section~\ref{sec:conclusion} concludes this
paper.

Compared to the conference
publication~\cite{DBLP:conf/lics/BiernackaBLPPS17},
Sections~\ref{sec:cbv} and~\ref{sec:control} are new, and the proofs
have been inlined in Sections~\ref{sec:full-abs}
and~\ref{sec:full-abstr-applicative}. 

\section{The Calculus \hocore}
\label{sec:hocore}

\paragraph*{Syntax and semantics} \hocore~\cite{Lanese-al:IC11} is
a simpler version of HO$\pi$~\cite{Sangiorgi:IaC96} where name restriction
is removed.  We let $a$, $b$, \etc range over channel names, and $x$,
$y$, \etc range over process variables. The syntax of \hocore
processes is as follows.
\[P, Q \bnf \inp a x P \midd \out a P \midd P \parallel Q \midd x \midd
\nil\]

The process $\inp a x P$ is waiting for a message on $a$ which, when
received, is substituted for the variable $x$ in $P$. If $x$ does not
occur in $P$, then we write $\inp a {\_} P$. The process $\out a P$ is
sending a message on $a$.  Note that communication is higher
order---processes are sent---and asynchronous---there is no
continuation after a message output. The parallel composition of
processes is written $P \parallel Q$, and the process $\nil$ cannot
perform any action. Input has a higher precedence than parallel
composition, e.g., we write $\inp a x P \parallel Q$ for
$(\inp a x P) \parallel Q$. We implicitly consider that parallel
composition is associative, commutative, and has $\nil$ as a neutral
element. In an input $\inp a x P$, the variable $x$ is bound in
$P$. We assume bound names to be pairwise distinct and distinct from
free names, using $\alpha$-conversion if necessary; we follow this
convention for all binding constructs in this paper. We write $\fn{P}$
for the free channel names of~$P$.

Informally, when an output $\out a P$ is in parallel with an input
$\inp a x Q$, a communication on $a$ takes place, producing
$\subst P x Q$, the capture avoiding substitution of $x$ by $P$ in
$Q$. We define in Figure~\ref{fig:lts} the semantics of \hocore as a
LTS, omitting the rules symmetric to \hname{Par} and \hname{Tau}. The
labels (ranged over by $l$) are either $\tau$ for internal
communication, $\out a P$ for message output, or $\inpp a P$ for
process input. We label $a$ an input where the received process does
not matter (e.g., $\inp a {\_ } P \trans a P$).
  
 \begin{figure}
  \begin{tabular}{c@{\hspace*{2cm}}c}
  \inferrule[\hname{Out}]{ }{\out a P \trans{\out a P} \nil}
  &
  \inferrule[\hname{Inp}]{ }{\inp a x Q \trans{\inpp a P} \subst P x Q}
  \\[1cm]
  \inferrule[\hname{Par}]{P\trans l P'}{P\parallel Q \trans l P'\parallel Q}
  &
  \inferrule[\hname{Tau}]{P\trans{\out a R} P'\and Q\trans{\inpp a R} Q'}{P\parallel Q \trans\tau P'\parallel Q'}
  \end{tabular}
   \caption{\hocore LTS}
   \label{fig:lts}
 \end{figure}

 Weak transitions allow for
internal actions before and after a visible one. We write $\wtrans\tau$ for the
reflexive and transitive closure $\rtc{\trans\tau}$, and $\wtrans l$ for
$\wtrans\tau \trans l \wtrans\tau$ when $l \neq \tau$.

\paragraph*{Barbed equivalence} Our definition of observable action
(or barb) and barbed equivalence depends on a finite set $\hidden$ of
hidden names, which allows some observable actions to be
ignored. Instead of adding top-level name restrictions on these names,
as in~\cite{Lanese-al:IC11}, we prefer to preserve the semantics of
the calculus and simply hide some names in the equivalence. Hidden
names are not a computational construct and are not required for the
encoding of the KAM, but they are necessary to protect the encoding
from an arbitrary environment when proving full abstraction. We
emphasize that we do not need the full power of name restriction: the
set of hidden names is finite and static---there is no way to create
new hidden names.

We let $\obssact$ range over names $a$ and \emph{conames} $\outn a$.
\begin{defi}
  The process $P$ has a strong observable action on $a$ (resp.
  $\outn a$) \wrt $\hidden$, written $P \obsh a \hidden$ (resp.
  $P \obsh {\outn a} \hidden$), if $a \notin \hidden$ and
  $P \trans{\inpp a Q} R$ (resp. $P \trans{\out a Q} R$) for some $Q$
  and $R$. A process~$P$ has a weak observable action on $\obssact$
  \wrt $\hidden$, written $P \wkobsh \obssact \hidden$, if
  $P \wtrans\tau P' \obsh \obssact \hidden$ for some~$P'$. We write
  $\WkObsh P \hidden$ for the set of weak observable actions of $P$
  \wrt $\hidden$.
\end{defi}
\begin{defi}
  A symmetric relation $\relR$ is a barbed bisimulation \wrt $\hidden$ if $P \relR Q$
  implies
  \begin{itemize}
  \item $P \obsh \obssact \hidden$ implies $Q \wkobsh \obssact \hidden$;
  \item for all $R$ such that $\fn{R} \cap \hidden = \emptyset$, we
    have $P {\parallel} R ~\relR~ Q {\parallel} R$;
  \item if $P \trans\tau P'$, then there exists $Q'$ such that $Q \wtrans\tau
    Q'$ and $P' \relR Q'$.
  \end{itemize}
  Barbed equivalence \wrt $\hidden$, noted $\hhobisim \hidden$, is the
  largest barbed bisimulation \wrt~$\hidden$.
\end{defi}
A strong barbed equivalence can be defined by replacing
$\wkobsh \obssact \hidden$ with $\obsh \obssact \hidden$ in the first
item, and $\wtrans\tau$ with $\trans\tau$ in the
third. From~\cite{Lanese-al:IC11}, we know that strong barbed
equivalence is decidable when all names can be observed
($\hidden = \emptyset$), but undecidable with 4 hidden names. We lower
this bound to 2 in Theorem~\ref{thm:undecidable}.

\section{Encoding the Krivine Abstract Machine}
\label{sec:KAM}

We show in this section that \hocore may faithfully encode a call-by-name
$\lambda$-calculus through an operationally equivalent encoding of the KAM.

\subsection{Definition of the KAM}

The KAM~\cite{Krivine:HOSC07} is a machine for call-by-name evaluation of closed
$\lambda$-calculus terms. We present a substitution-based variant of the KAM for
simplicity, and to reuse the substitution of \hocore in the translation. A
\emph{configuration} $C$ of the machine is composed of the term $t$ being
evaluated, and a stack $\stk$ of $\lambda$-terms. Their syntax and the
transitions are as follows.

\begin{align*}
  C & \bnf \kam{t}{\stk} \tag{configurations}\\
  t,s &\bnf x \midd t \app s \midd \lam x t \tag{terms}\\
  \stk &\bnf t \cons \stk \midd \mts\tag{stacks} \\
  \\[-3mm]
  \kam{t \app s}{\stk} &\kred \kam{t}{s \cons \stk} \tag{\hname{Push}}\\
  \kam{\lam x t}{s \cons \stk} &\kred \kam{\subst{s}{x}{t}}{\stk} \tag{\hname{Grab}}
\end{align*}
\noindent A $\lambda$-abstraction $\lam x t$ binds $x$ in $t$; a term
is closed if it does not contain any free variables. We use $\mts$ to
denote the empty stack. In $\hname{Push}$, the argument $s$ of an
application is stored on the stack while the term $t$ in function
position is evaluated. If we get a $\lambda$-abstraction $\lam x t$,
then an argument $s$ is fetched from the stack (transition
$\hname{Grab}$), and the evaluation continues with $\subst s x t$, the
capture-avoiding substitution of $x$ by $s$ in $t$. If a configuration
of the form $\kam {\lam x t} \mts$ is reached, then the evaluation is
finished, and the result is $\lam x t$. Because we evaluate closed
terms only, it is not possible to obtain a configuration of the form
$\kam x \stk$.

\subsection{Translation into \hocore}

The translation of the KAM depends essentially on how we push and grab terms on
the stack. We represent the stack by two messages, one on name $\hp$ for its
head, and one on name~$c$ (for \emph{continuation}) for its
tail (henceforth, a stack $q$ is always encoded as a message on 
  $hd_{\mathsf q}$ for its head and one on $q$ for its tail). The empty
stack can be represented by an arbitrary, non-diverging, deterministic process,
e.g.,~$\nil$; here we use a third name to signal that the computation is
finished with $\out b \nil$. As an example, the stack
$1 \cons 2 \cons 3 \cons 4 \cons \mts$ is represented by
$\out \hp 1 \parallel \out c {\out \hp 2 \parallel \out c {\out \hp 3 \parallel
    \out c {\out \hp 4 \parallel \out c {\out b \nil}}}}$.

With this representation, pushing an element $e$ on a stack $p$ is done by
creating the process $\out \hp e \parallel \out c p$, while grabbing the head of
the stack corresponds to receiving on $\hp$. With this idea in mind, we define the
translations of the entities defining the KAM, starting with stacks.
\begin{align*}
  \tr{[]} &\eqdef \out b \nil \\
  \tr{t \cons \stk} &\eqdef \out{\hp}{\tr{t}} \parallel \out{c}{\tr{\stk}}
\end{align*}
In the translation of a configuration $\kam{t}{\stk}$, we reuse the name~$c$ to
store the stack, meaning that before pushing on $\stk$ or grabbing the head of
$\stk$, we have to get $\tr\stk$ by receiving on $c$.
\begin{equation*}
  \tr{\kam{t}{\stk}} \eqdef \tr{t} \parallel \out{c}{\tr{\stk}}
\end{equation*}
For instance, in the application case $\tr{t \app s}$, we start by
receiving the current stack $p$ on $c$, and we then run $\tr t$ in
parallel with the translation of the new stack
$\out \hp {\tr s} \parallel \out c p$. Henceforth, we assume the
variable $p$ does not occur free in the translated entities.
\begin{equation*}
  \tr{t \app s}     \eqdef \inp{c}{p}{\perfectparens{\tr{t} \parallel
                      \out{c}{\out{\hp}{\tr{s}} \parallel \out{c}{p}}}}
\end{equation*}
Similarly, in the $\lambda$-abstraction
case $\tr {\lam x t}$, we get the current stack $p$ on $c$, that we run in
parallel with $\inp \hp x {\tr t}$. If $p$ is not empty, then it is a process of
the form $\out \hp {\tr s} \parallel \out c {\tr{\stk}}$, and a communication on
$\hp$ is possible, realizing the substitution of $x$ by $s$ in $t$; the
execution then continues with $\tr{\subst s x t} \parallel \out c {\tr\stk}$.
Otherwise,~$p$ is $\out b \nil$, and the computation terminates.
\begin{align*}
  \tr{\lam x t} &\eqdef \inp{c}{p}{(\inp{\hp}{x}{\tr{t}} \parallel p)}\\
  \tr{x}        &\eqdef x
\end{align*}
Formally, the operational correspondence between the KAM and its translation is
as follows.
\begin{thm}
  \label{th:cor-KAM-trans}
  In the forward direction, if $C \kred^{*} C'$, then $\tr C \wtrans\tau
  \tr{C'}$. In the backward direction, if $\tr C \wtrans\tau P$, then there
  exists a $C'$ such that $C \kred^{*} C'$ and either
    \begin{itemize}
    \item $P = \tr{C'}$,
    \item or there exists $P'$ such that $P \trans\tau P' = \tr{C'}$,
    \item or $C' = \kam {\lam x t} \mts$ and $P = \inp{\hp}{x}{\tr{t}} \parallel
      \out b \nil$.
    \end{itemize}
\end{thm}

\begin{proof}
  The proof is straightforward in the forward direction. In the
  backward direction, we show that the translation is deterministic
  (if $\tr C \wtrans\tau P \trans\tau Q_1$ and $\tr C \wtrans\tau P
  \trans\tau Q_2$, then $Q_1 = Q_2$) and we rely on the fact that the
  translation of a \hname{Push} step uses one communication, while we use two
  communications for a \hname{Grab} step.
\end{proof}

\noindent A direct consequence is we can observe on the \hocore side when the
reduction of a $\lambda$-term converges.

\begin{cor}
  \label{cor:converges-obs}
  $C \kred^{*} \kam {\lam x t} \mts$ for some $t$ iff $\tr{C} \wkobsh
  {\outn b} \emptyset$.
\end{cor}

\begin{proof}
  If $C \kred^{*} \kam {\lam x t} \mts$ for some $t$, then by
  Theorem~\ref{th:cor-KAM-trans},
  $\tr C \wtrans \tau \tr{\kam {\lam x t} \mts} \trans\tau \inp{\hp}{x}{\tr{t}}
  \parallel \out b \nil$, and the result holds.

  Conversely, suppose $\tr{C} \wkobsh {\outn b} \emptyset$: there exists $P$ such that $\tr
  C \wtrans\tau P \obsh{\outn b} \emptyset$. One can check that in the first two cases of the
  backward direction of Theorem~\ref{th:cor-KAM-trans}, $P$ cannot exhibit $\outn b$
  as an observable action, so we are in the last case, as wished. 
\end{proof}

\noindent We can then improve over the result of~\cite{Lanese-al:IC11} about
undecidability of strong barbed equivalence by hiding $\hp$ and $c$.

\begin{thm}
  \label{thm:undecidable}
  Strong barbed equivalence is undecidable in \hocore with 2 hidden names.
\end{thm}

\begin{proof}
  Assume we can decide strong barbed equivalence with two hidden names. Let $t$
  be a closed $\lambda$-term and
  $P_\Omega \eqdef \inp{c}{x}{\perfectparens{x \parallel \out c x}} \parallel
  \out c {\inp{c}{x}{\perfectparens{x \parallel \out c x}}}$. We can thus decide
  if $\tr{\kam t \mts}$ is strong barbed-equivalent to $P_\Omega$ when $\hp$ and
  $c$ are hidden, but $P_\Omega$ loops with no observable action, and $t$
  converges iff $\tr{\kam t \mts}$ has a weak observable action on $\outn b$ by
  Corollary~\ref{cor:converges-obs}. By deciding barbed equivalence, we can
  therefore decide whether the reduction of~$t$ converges, hence a contradiction.
\end{proof}

\section{Normal-Form Bisimilarity}
\label{sec:full-abs}

Our first full abstraction result is for normal-form
bisimilarity~\cite{Lassen:MFPS99}. We show how to internalize this equivalence
in an extension of the KAM such that it may be captured by a simple barbed
bisimilarity. We then translate this extended KAM into \hocore, and we finally
prove full abstraction.

\subsection{Normal-Form Bisimilarity}

Normal-form bisimilarity compares terms by reducing them to weak head normal
forms, if they converge, and then decomposes these normal forms into subterms
that must be bisimilar. Unlike the KAM, normal-form bisimilarity is defined on
open terms, thus we distinguish free variables, ranged over by $\fvr$, from
bound variables, ranged over by $x$. The grammars of terms ($t,s$) and values
($v$) become as follows.
\[
  t, s \bnf  \fvr \midd x \midd \lam x t \midd t \app s \qquad v \bnf  \fvr
  \midd \lam x t 
\]

Henceforth, we assume that $\lambda$-terms are well
formed, i.e., all variables ranged over by~$x$ are bound: $x$ is not a valid
term but $\fvr$ is. We write $\fv t$ for the set of free variables of $t$. A
variable $\fvr$ is said \emph{fresh} if it does not occur in any of the entities
under consideration. 

When evaluating an open term, we can obtain either a
$\lambda$-abstraction with an empty stack, or a
free variable and a stack. We inductively extend a relation $\relR$ on
$\lambda$-terms to stacks by writing $\stk_{1} \relR \stk_{2}$ if
$\stk_{1} = \stk_{2} = \mts$, or if $\stk_1 = t \cons \stk'_1$,
$\stk_2 = s \cons \stk'_2$, $t \relR s$, and $\stk'_1 \relR \stk'_2$.

\begin{defi}
  \label{def:nf-cbn}
  A symmetric relation $\relR$ is a normal-form bisimulation if $t \relR s$ implies:
  \begin{itemize}
  \item if $\kam t \mts \rtc\kred \kam {\lam x {t'}} \mts$, then there exists
    $s'$ such that $\kam s \mts \rtc\kred \kam{\lam x {s'}} \mts$ and
    $\subst \fvr x {t'} \relR \subst \fvr x {s'}$ for a fresh $\fvr$;
  \item if $\kam t \mts \rtc\kred \kam \fvr \stk$, then there exists $\stk'$
    such that $\kam s \mts \rtc\kred \kam \fvr {\stk'}$ and $\stk \relR \stk'$.
  \end{itemize}
  Normal-form bisimilarity $\nfbisim$ is the largest normal-form bisimulation. 
\end{defi}
Normal-form bisimilarity is not complete \wrt contextual equivalence
in $\lambda$-calculus: there exists contextually equivalent terms
which are not normal-form bisimilar~\cite[Example
5.5]{Sangiorgi:LICS92}. In this work, we evaluate terms to weak-head
normal forms, and the resulting bisimilarity characterizes
L\'evy-Longo tree
equivalence~\cite{Levy:LCCST75,Ong:PhD88,Sangiorgi:LICS92}. If we were
to evaluate to head normal forms instead, it would characterize B\"ohm
tree equivalence~\cite{Barendregt:84,Lassen:MFPS99}.

\subsection{Abstract Machine}
\label{ssec:am-cbn}

We extend the KAM so that it provides additional steps, identified by
labeled transitions, that capture the testing done by normal-form
bisimilarity. Roughly, the extended machine explores the L\'evy-Longo
tree of a $\lambda$-term, choosing arbitrarily a branch when several
are available, but signaling these choices to the outside so that they
can be mimicked.  To do so, we rely on \emph{flagged transitions},
\emph{terminating transitions}, and a restricted form of
\emph{non-determinism}.

Flagged transitions are usual transitions of the machine with some
additional information to convey to the environment that a particular
event is taking place. Machine bisimilarity, defined below, ensures
that bisimilar machines have matching flags. Transitions without flags
use a $\tau$ label. Terminating transitions are flagged transitions
that indicate the computation has stopped.  They are needed for
bisimilarity: as machine bisimilarity ignores $\tau$ labels, we use
terminating transitions to distinguish between terminated and
divergent machine runs. Finally, we allow some non-determinism in machines,
i.e., a given configuration may take two different transitions, only
if the transitions are flagged and have different flags. In other
words, the non-deterministic choice is made explicit to the
environment.

\begin{figure}[t]
\begin{align*}
\fneval {t \app s} \stk n & \fnred{\tau}  \fneval t {s \cons \stk} n \tag{\hname{Push}}\\
\fneval {\lam x t}{s \cons \stk } n & \fnred{\tau}  \fneval {\subst s x t} \stk n \tag{\hname{Grab}}\\
\fneval {\lam x t} \mts n & \fnred{\flambda} \fneval {\subst n x t} \mts
{n+1} \tag{\hname{Lambda}}\\
\fneval \fvr \stk n & \fnred{\fvr} \fncont \stk n \tag{\hname{Var}}
\\[3mm]
\fncont \mts n & \fnred{\fdone} \tag{\hname{Done}}\\
\fncont {t \cons \stk} n & \fnred{\fenter} \fneval t \mts n \tag{\hname{Enter}}\\
\fncont {t \cons \stk} n & \fnred{\fskip}  \fncont \stk n \tag{\hname{Skip}} 
\end{align*}
\caption{NFB Machine}
\label{fig:machine}
\end{figure}

We define the NFB machine in Figure~\ref{fig:machine}. When computation stops
with a $\lambda$-abstraction and an empty stack, we have to restart the machine
to evaluate the body of the abstraction with a freshly generated free variable
(rule \hname{Lambda}). To do so, we consider free variables as natural numbers,
and we keep a counter~$n$ in the machine which is incremented each time a fresh
variable is needed. For a configuration $\kam \fvr \stk$, normal-form
bisimilarity evaluates each of the $t_i$ in the stack (rule \hname{Var}). To
internalize this step, we could launch several machines in parallel, as
in~\cite{vanBakel-Vigliotti:CLC14}, where the translated $t_i$ are run in
parallel. This approach has two drawbacks: first, it is a further extension of
abstract machines---a machine no longer steps to a single machine state but to a
multiset of states. Second, when translating such extended machines into
\hocore, we want to prevent them from interacting with each other, but we cannot
rely on name restriction, as in~\cite{vanBakel-Vigliotti:CLC14}, to
encapsulate an unbounded number of translations.  Alternatively, one could evaluate
the elements of the stack sequentially, but this approach fails if one of the
elements of the stack diverges, as the later elements will never be
evaluated.

We thus consider a third approach, built upon flagged non-determinism:
the machine chooses arbitrarily an element of the stack to evaluate,
and signals this choice using flags (rules \hname{Enter},
\hname{Skip}, and \hname{Done}). The burden of evaluating every
element of the stack is thus relegated to the definition of machine
bisimilarity: as every flagged execution must be matched by an
execution with the same flags, every possible choice is explored.  It
is reminiscent of how operational game
semantics~\cite{Levy-Staton:CSLLICS14,Jaber:FOSSACS15} equates terms:
the possible interactions of a program with its environment are
represented with labeled transitions, the execution of a program is a
trace of such labels, and trace equivalence ensures that every
possible execution has been considered.

As before, we use $\sts$ to range over configurations, which are now
of two kinds. In \emph{evaluation mode}, $\fneval t \stk n$ is
reducing~$t$ within stack $\stk$ and with counter~$n$. The transitions
$\hname{Push}$ and $\hname{Grab}$ are as in the KAM, except for the
extra parameter. If we reach a $\lambda$-abstraction in the empty
context (transition $\hname{Lambda}$), then the machine flags
$\flambda$ and then restarts to evaluate the body, replacing the bound
variable by a fresh free variable, i.e., the current value $n$ of the
counter. If we reach a free variable~$\fvr$, i.e., a number, then we
flag $\fvr$ before entering the next mode (transition $\hname{Var}$).

In \emph{continuation mode} $\fncont \stk n$, the transition
$\hname{Done}$ simply finishes the execution if $\stk = \mts$, using
the flag $\fdone$. Otherwise, $\stk = t \cons \stk'$, and the machine
either evaluates~$t$ with flag $\fenter$ (and forgets about~$\stk'$),
or skips $t$ with a flag $\fskip$ to possibly evaluate a term in
$\stk'$. The machine may skip the evaluation of all the terms
in~$\stk$, but it would still provide some information, as it would
generate~$m$~$\fskip$ messages (followed by $\fdone$), telling us that
$\stk$ has $m$ elements. Note that the counter~$n$ is stored in
continuation mode just to be passed to the evaluation mode when one of
the $t_i$ is chosen with transition $\hname{Enter}$.

\begin{exa}
  To illustrate how the machine works, we show the transitions starting from the
  term $(\lam x x) \app (\lam y {y \app 0 \app \Omega})$, where
  $\Omega \eqdef (\lam x {x \app x})\app (\lam x {x \app x})$. The term is
  executed in the empty context, and with a counter initialized to a value
  greater than its free variables.
  \begin{align*}
    & \fneval {(\lam x x) \app (\lam y {y \app 0 \app \Omega})} \mts 1 \\
    \fnred{\tau}\; &
    \fneval {\lam x x}{\lam y {y \app 0 \app \Omega} \cons \mts} 1
    \tag{\hname{Push}} \\
     \fnred{\tau}\; &\fneval {\lam y {y \app 0 \app \Omega}} \mts 1
    \tag{\hname{Grab}} \\
     \fnred{\flambda}\; &\fneval {1 \app 0 \app \Omega} \mts 2 \tag{\hname{Lambda}} \\
     \fnred{\tau}\; &\fneval {1 \app 0}{\Omega \cons \mts} 2 \fnred{\tau}
                      \fneval 1 {0 \cons \Omega \cons \mts
      } 2 \tag{\hname{Push} - \hname{Push}}\\
     \fnred{1}\; &\fncont{0 \cons \Omega \cons \mts} 2 \tag{\hname{Var}}
  \end{align*}
  We then have three possibilities. First, we reduce the top of the
  stack, with the sequence
  \( \fncont{0 \cons \Omega \cons \mts} 2 \fnred{\fenter} \fneval 0
  \mts 2 \fnred{0} \fncont \mts 2 \fnred{\fdone} \).  Second, we
  evaluate~$\Omega$ with the sequence
  \( \fncont{0 \cons \Omega \cons \mts} 2 \fnred{\fskip} \fncont
  {\Omega \cons \mts} 2 \fnred{\fenter} \fneval \Omega \mts 2, \) and
  then the machine loops without generating any flag. Third, we skip
  both terms with
  \( \fncont{0 \cons \Omega \cons \mts} 2 \fnred{\fskip} \fncont
  {\Omega \cons \mts} 2 \fnred{\fskip} \fncont \mts 2 \fnred{\fdone}.
  \) Note that the three options generate different traces of flags.

\end{exa}

Because the rules $\hname{Grab}$ and $\hname{Push}$ are the same between the KAM
and the NFB machine, there is a direct correspondence between the two. 
\begin{lem}
  \label{lem:kam-nfam}
  For all $t$, $t'$, $\stk$, $\stk'$, $n$, $\kam t \stk \kred \kam {t'}{\stk'}$
  iff $\fneval t \stk n \fnred\tau \fneval {t'}{\stk'} n$.
\end{lem}

We finally show that a notion of bisimilarity between configurations of an NFB
machine captures normal-form bisimilarity. To this end, we first define machine
bisimilarity, where we denote the flags of the machine by $\flag$.

\begin{defi}
  A symmetric relation $\relR$ is a machine bisimulation if $C_1 \relR C_2$
  implies:
  \begin{itemize}
  \item if $C_1 \fnred{\tau}^{*}\fnred{\flag} C'_1$, then there exists $C'_2$ such that $C_2
    \fnred{\tau}^{*}\fnred{\flag} C'_2$ and $C'_1 \relR C'_2$;
  \item if $C_1 \fnred{\tau}^{*}\fnred{\flag}$, then
    $C_2 \fnred{\tau}^{*}\fnred{\flag}$.
  \end{itemize}
  Machine bisimilarity $\ambisim$ is the largest machine bisimulation.
\end{defi}

Intuitively, machine bisimilarity ensures that every flag emitted by a machine
is matched by an identical flag from the other machine, up to internal
reductions. Note that a machine that diverges with $\tau$ labels can be related
to any other diverging machine or any machine stuck without a flag. We make sure
the latter case cannot occur in our machines by having only terminating
transitions, which are flagged, as stuck transitions. We can now state that
normal-form bisimilarity coincides with machine bisimilarity of NFB machines.

\begin{thm}
  \label{th:nf-am}
  $t \nfbisim s$ iff there exists $n > \max{(\fv t \cup \fv s)}$ such that
  $\fneval t \mts n \ambisim \fneval s \mts n$.  
\end{thm}

\begin{proof}
  To prove that machine bisimilarity implies normal-form bisimilarity, we show
  that
  \[ \mathord{\relR} \eqdef \{ (t, s) \midd \fneval t \mts n \ambisim
    \fneval s \mts n, n > \max{(\fv t \cup \fv s)}\} \] is a
  normal-form bisimulation. Suppose
  $\kam t \mts \rtc\kred \kam {\lam x {t'}} \mts$. By
  Lemma~\ref{lem:kam-nfam}, we have
  $\fneval t \mts n \fnredrtc\tau \fneval {\lam x {t'}} \mts n$, which
  can only be followed by the transition
  $\fneval {\lam x {t'}} \mts n \fnred\flambda \fneval {\subst n x
    {t'}} \mts {n+1}$.  Because
  $\fneval t \mts n \ambisim \fneval s \mts n$, there exists a
  configuration $\sts'$ such that
  $\fneval s \mts n \fnredrtc\tau \fnred\flambda \sts'$ and
  $\fneval {\subst n x {t'}} \mts {n+1} \ambisim \sts'$. The
  $\tau$-steps cannot change the value of $n$, and the flag $\flambda$
  ensures that there exists $s'$ such that
  $\fneval s \mts n \fnredrtc\tau \fneval {\lam x {s'}} \mts n
  \fnred\flambda C'$ with $C' = \fneval {\subst n x {s'}} \mts
  {n+1}$. Using Lemma~\ref{lem:kam-nfam} again, we deduce
  $\kam s \mts \rtc\kred \kam {\lam x {s'}} \mts$ with
  $\fneval {\subst n x {t'}} \mts {n+1} \relR \fneval {\subst n x
    {s'}} \mts {n+1}$, as wished. The case
  $\kam t \mts \rtc\kred \kam \fvr \stk$ is similar.

  For the reverse implication, we show that
  \begin{align*}
    \relR & \eqdef \{ (\fneval t \mts n, \fneval s \mts n) \midd t \nfbisim s, n >
            \max{(\fv t \cup \fv s)} \} \\
          &\cup \{ (\fncont \stk n, \fncont {\stk'} n) \midd
            \stk \nfbisim \stk', n > \max{(\fv\stk \cup \fv{\stk'})} \}
  \end{align*}
  is a machine bisimulation. Suppose
  $\fneval t \mts n \fnredrtc\tau \fnred\flambda \fneval {\subst n x
    {t'}} \mts {n+1}$. By Lemma~\ref{lem:kam-nfam}, we have
  $\kam t \mts \rtc\kred \kam {\lam x {t'}} \mts$. Because
  $t \nfbisim s$, there exists $s'$ such that
  $\kam s \mts \rtc\kred \kam {\lam x {s'}} \mts$ and
  $\subst n x {t'} \nfbisim \subst n x {s'}$. From
  Lemma~\ref{lem:kam-nfam}, we deduce
  $\fneval s \mts n \fnredrtc\tau \fnred\flambda \fneval {\subst n x
    {s'}} \mts {n+1}$, with
  $ \fneval {\subst n x {t'}} \mts {n+1} \relR \fneval {\subst n x
    {s'}} \mts {n+1}$, as wished. The other cases are similar. 
\end{proof}

\subsection{Translation into \hocore}
\label{ss:nfb-hocore}

\begin{figure}[t]
\begin{align*}
  \tr{t \app s} & \eqdef  \inp{c}{p}{\perfectparens{\tr{t} \parallel \out{c}{\out{\hp}{\tr{s}} \parallel \out{c}{p}}}}\\
  \tr{\lam x t} & \eqdef  \inp{c}{p}{\perfectparens{p \parallel \out b {\Restart} \parallel \inp{\hp}{x}{\inp b {\_}{\tr{t}}}}}\\
  \tr{x} & \eqdef x\\
  \tr{\fvr} &\eqdef \trint \fvr \\[2mm]
  \Restart & \eqdef  \inp{\flambda}{\_}{\inp k x}{(
               \out \hp x \parallel \out k {\inp {suc}{\_} x} \parallel 
               \out c {\tr\mts} \parallel \out b \nil
             )}\\
  \Rec  &\eqdef \inp{init}{\_}{\inp{rec} x {\perfectparens{x \parallel \out{rec
      } x \parallel \Contt}}} \\
  \Contt & \eqdef \inp c p {\perfectparens{p \parallel \out b {\inp {\fdone}{\_} {\nil}} \parallel 
    \inp \hp x {\inp b {\_}{\IChoice x}}}} \\
  \IChoice{P_t} & \eqdef   \inp {\fenter}{\_}{\inp c {\_} {\perfectparens{P_t \parallel \out
              c {\tr \mts}}}} +
        \inp{\fskip}{\_}{\out{init}{\nil}}
  \\[2mm]
  \tr{\mts} & \eqdef \inp b x x\\
  \tr{t \cons \stk} &\eqdef \out{\hp}{\tr{t}} \parallel \out{c}{\tr{\stk}}\\
  \trint 0   & \eqdef \inp z {\_}{\out{init} \nil} \\
  \trint {n+1}  & \eqdef \inp {suc}{\_}{\trint{n}} 
  \\[2mm]
  \tr{\fneval t \stk n} &\eqdef \tr{t} \parallel \out{c}{\tr \stk} \parallel \out k
  {\trint n} \parallel \Rec \parallel \out{rec}{\Rec} \\
  \tr{\fncont \stk n} &\eqdef \out{c}{\tr \stk} \parallel \out k
  {\trint n} \parallel \Contt \parallel \Rec \parallel \out{rec}{\Rec}
\end{align*}
\caption{Translation of the NFB machine into \hocore}
\label{fig:translation}
\end{figure}

In Figure~\ref{fig:translation}, we present the translation of the NFB
machine into \hocore, where we consider flags as channel
names. Configurations now contain a counter~$n$, which is represented
by a message on $k$ containing the value of~$n$ encoded as a
process. We use $\trint .$ to translate a natural number $n$ into a
process
$\underbrace{suc(\_).\ldots suc(\_)}_{n \text{ times}}.\inp z
{\_}{\out{init} \nil}$; the role of the final output on $init$ is
explained later. Since we cannot generate new names in \hocore, free
variables~$\fvr$ are also numbers, and we use the same translation for
them. We also use non-deterministic internal choice, encoded as
follows:
$P+Q \eqdef \out{ch} P \parallel \out {ch} Q \parallel \inp{ch} x
{\inp {ch} \osef x}$: both messages are consumed, and only one process
is executed. This encoding supposes that at most one choice is active
at a given time, as we use only one name $ch$ to encode all the
choices. We also use $n$-ary choices for $n > 2$ in
Section~\ref{ssec:trans-AM-AB}, which can be encoded in the same way.

A stack is represented as in the KAM, by messages on $\hp$ and $c$, and the
translation of an application $\tr{t \app s}$ is exactly the same as for the
KAM. The encoding of the empty context~$\mts$ is different, however, because
contexts are used to distinguish between execution paths at two points in the
machine: when evaluating a function $\lam x t$ in evaluation mode, and when
deciding whether the execution is finished in continuation mode. The empty context is
thus encoded as $\inp b x x$, waiting to receive the process to execute in the
empty case. For the non-empty case, this input on $b$ is absent and there are
instead messages on $\hp$ and~$c$. Thus the generic way to choose a branch is as
follows:
\begin{equation*}
  \out b {\text{do this if empty}} \parallel \inp \hp x {\inp c y {\inp b {\_} {\text{do this if non-empty}}}}.
\end{equation*}
In the non-empty case, the input on $b$ discards the message for the empty
behavior that was not used.

For $\lambda$-abstractions, the behavior for the empty case is described in the
process $\Restart$. More precisely, $\tr{\lam x t}$ receives the current stack
$\tr \stk$ on $c$ to run it in parallel with
$\out b {\Restart} \parallel \inp{\hp}{x}{\inp b {\_}{\tr{t}}}$. If $\tr \stk$
is of the form $\out \hp {\tr{t'}} \parallel {\out c {\tr {\stk'}}}$, then we
have the same behavior as with the KAM, with an extra
communication on $b$ to garbage collect the $\Restart$ process.  Otherwise,
$\tr \stk = \inp b x x$ and we obtain the following sequence of transitions.
\begin{align*}
  & \inp b x x \parallel \out b {\Restart} \parallel \inp{\hp}{x}{\inp b {\_}{\tr{t}}} \parallel \out k {\trint n} \\
  \trans\tau \; & \inp{\flambda}{\_}{\inp k x}{\perfectparens{\out \hp x \parallel \out k
    {\inp {suc}{\_} x} \parallel \out c {\tr\mts} \parallel \out b
    \nil}} \parallel \inp{\hp}{x}{\inp b
    {\_}{\tr{t}}} \parallel \out k {\trint n} \\
  \wtrans{\flambda} \; & \out \hp {\trint n} \parallel \out k
  {\inp {suc}{\_}{\trint n}} \parallel \out c {\tr\mts} \parallel \out b
  \nil \parallel \inp{\hp}{x}{\inp b {\_}{\tr{t}}}\\
  \trans{\tau} \; & \out k
  {\inp {suc}{\_}{\trint n}} \parallel \out c {\tr\mts} \parallel \out b
  \nil \parallel \inp b {\_}{\tr{\subst n x t}}  \\
  \trans{\tau} \; & \out k
  {\trint {n+1}} \parallel \out c {\tr\mts} \parallel \tr{\subst n x t}
\end{align*}
In the end, we have effectively restarted the machine to evaluate $\subst n x
t$, as wished.

In continuation mode, the branching is done by the process $\Contt$,
which is executed after applying the transition \hname{Var}. More
precisely, a free variable $\fvr$ is translated using~$\trint{.}$,
which signals first the value of $\fvr$ (with the names $suc$
and~$z$), and then sends a message on $init$ to enter the continuation
mode. The way the NFB machine chooses which $t_i$ to evaluate in a
stack $t_1 \cons \ldots \cons t_m \cons \mts$ is a recursive
mechanism, and recursion can be encoded in a higher-order calculus:
$\Rec \parallel \out {rec} \Rec$ reduces to
$\Contt \parallel \Rec \parallel \out {rec} \Rec$ when it receives a
message on $init$. The process $\Contt$ is doing a case analysis on
$\tr\stk$ when it is executed in parallel with $\out c {\tr \stk}$: if
$\stk = \mts$, then $\tr \stk = \inp b x x$ receives the message on
$b$ which flags~$\fdone$ and the machine stops. Otherwise,
$\tr \stk = \out \hp {\tr t} \parallel \out c {\tr {\stk'}}$, and we
have the following reductions:
\begin{align*}
  \Contt \parallel \out c {\tr \stk} &\trans\tau \out \hp {\tr t} \parallel \out c
  {\tr {\stk'}} \parallel \out b {\inp {\fdone}{\_} \nil} \parallel \inp \hp x {\inp b
    {\_} \IChoice x} \\ 
  &\mathbin{{\trans\tau}^2}  \out c {\tr {\stk'}} \parallel \IChoice {\tr t}
\end{align*}
At this point, $\IChoice {\tr t}$ either evaluates~$t$ with flag $\fenter$, or
flags~$\fskip$ and continues exploring~$\stk'$. In the former case, the current
stack~$\stk'$ is replaced by an empty stack, and in the latter, a message on
$init$ is issued to produce $\Contt \parallel \out c {\tr{\stk'}}$ after some
reduction steps.

\subsection{Operational Correspondence and Full Abstraction}
\label{ssec:fullabs-nf}

Establishing full abstraction requires first to state the
correspondence between the NFB machine and its translation. We call
\emph{machine process} a process obtained by reducing a translated
configuration.
\begin{defi}
  \label{def:machine-proc}
  A process $P$ is a machine process if there exists a configuration $\sts$ of
  the machine such that $\tr\sts \wtrans\tau P$. 
\end{defi}

We recall that $\flag$ range over the flags of the NFB machine, i.e.,
$\flambda$, $\fenter$, $\fskip$, $\fdone$, and~$\fvr$. We let
$\fortau$ range over flags and $\tau$. The translation maps the flags
to corresponding channel names, except for $\fvr$, which is
represented using the $suc$ and $z$ channel names. We call \emph{flag
  names} the channel names $\flambda$, $\fenter$, $\fskip$, $\fdone$,
$suc$, and~$z$. We write $\trans \fvr$ as a shorthand for $\fvr$
transitions $\trans {suc}$ followed by a transition $\trans z$, i.e.,
$(\trans{suc})^\fvr \trans z$, and $\wtrans \fvr$ stands for
$\wtrans\tau \trans \fvr \wtrans\tau$. We define $\hidden$ as the set
of names in the translation that are not flag names. We write
$\obss \obssact$ as a shorthand for $\obsh \obssact \hidden$,
$\wkobs \obssact$ for $\wkobsh \obssact \hidden$, $\WkObs P$ for
$\WkObsh P \hidden$, and $\hobisim$ for $\hhobisim \hidden$.

Given a machine process $P$ and a flag $\fortau$, we define a set
$\nexts \fortau P$ such that $Q \in \nexts \fortau P$ implies there is
a weak reduction between $P$ and~$Q$ with at least one transition
labeled $\fortau$, $Q$ is the translation of a configuration, and
there is no other machine translation in between.  Intuitively, it is
the set of translations that are reachable from $P$ just after a step
$\fortau$.

\begin{defi}
  \label{def:next-config}
  We define $\nexts \fortau P$ so that $Q \in \nexts \fortau P$
  implies $P \wtrans\tau \trans \fortau \wtrans\tau Q$,
  $Q = \tr {\sts'}$ for some $\sts'$, and for any~$P'$ such that
  $P' \neq P$, $P' \neq Q$, and $P \wtrans \fortau P' \wtrans\tau Q$
  or $P \wtrans \tau P' \wtrans \fortau Q$, we have
  $P' \neq \tr{\sts}$ for any $\sts$.
\end{defi}
The translation has been designed so that each machine process is
either stuck waiting for a message on a flag name, or has at most one
possible communication, except when a non-deterministic choice has to
be made. The latter case occurs when the NFB machine is exploring a
stack, with a translated configuration of the form
$\tr{\fncont{t \cons \stk} n}$: as explained at the end of
Section~\ref{ss:nfb-hocore}, we have
\[
  \tr{\fncont{t \cons \stk} n} \wtrans\tau \IChoice{\tr t} \parallel \out c
  {\tr \stk} \parallel \out k {\tr n} \parallel \Rec \parallel \out{rec}
  \Rec 
\]
We call the resulting process
$\Choice{t}{\stk}{n} \eqdef \IChoice{\tr t} \parallel \out c {\tr
  \stk} \parallel \out k {\tr n} \parallel \Rec \parallel \out{rec}
\Rec$ a \emph{choice process}, as it is about to make an arbitrary
choice in $\IChoice{\tr t}$. Unfolding the encoding of choice, we have
$\IChoice{\tr t} = \out {ch}{P_{\fenter, t}} \parallel
\out{ch}{P_\fskip} \parallel \inp{ch}{x}{\inp{ch}{\_} x}$, with
$P_{\fenter, t} \eqdef \inp {\fenter}{\_}{\inp c {\_}
  {\perfectparens{\tr t \parallel \out c {\tr \mts}}}}$ and
$P_\fskip \eqdef \inp{\fskip}{\_}{\out{init}{\nil}}$. Let
$P_{\stk, n} \eqdef \out c {\tr \stk} \parallel \out k {\tr n}
\parallel \Rec \parallel \out{rec} \Rec$. Then the two possible
transitions from $\Choice{t}{\stk}{n}$ lead to the following machine
translations:
\begin{align*}
  \Choice{t}{\stk}{n}
  & \trans\tau \out{ch}{P_\fskip} \parallel \inp{ch}{\_}{P_{\fenter,
    t}} \parallel P_{\stk,n} \trans \tau P_{\fenter, t} \parallel
    P_{\stk,n} \trans \fenter {\trans\tau}^2 \tr{\fneval t \mts n}\\
  \Choice{t}{\stk}{n}
  & \trans\tau \out{ch}{P_{\fenter,t}} \parallel
    \inp{ch}{\_}{P_\fskip} \parallel P_{\stk,n} \trans \tau
    P_{\fskip} \parallel 
    P_{\stk,n} \trans \fskip {\trans\tau}^2 \tr{\fncont \stk n}
\end{align*}
The first transition $\trans\tau$ is making the choice. From there,
one process flags $\fenter$, while the other flags $\fskip$: we know
which choice has been made by looking at the weak observable
actions. Note that $\nexts \tau {\Choice{t}{\stk}{n}}$ is empty,
because we need to flag either $\fenter$ or $\fskip$ to reach a
machine translation.

The next lemma expresses the fact that the translation is
deterministic, relying on flags in the case of choice processes. 
\begin{lem}[Determinism]
  \label{lem:determinist}
  Let $P$ be a machine process.
  \begin{itemize}
  \item If~$P$ is not a choice process, then either $P$ cannot reduce,
    or there exist a unique $P'$ and $\fortau$ such that
    $P \trans \fortau P'$;
  \item if $P$ is a choice process, then there exist $P_1$ and $P_2$
    such that $P \trans\tau P_1$, $P \trans\tau P_2$, and
    $\WkObs{P_1} \neq \WkObs{P_2}$. For all $\fortau$, $P'$ such that
    $P \trans\fortau P'$, we have $\fortau = \tau$ and either $P' =
    P_1$ or $P' = P_2$.
  \end{itemize}
\end{lem}

\begin{proof}
  The proof is a straightforward case analysis on the shape of the
  machine processes, checking that for each of them, the lemma holds.
\end{proof}

\noindent Lemma~\ref{lem:determinist} implies that $\nexts \fortau P$ is either
empty or a singleton for every $\fortau$ and $P$. In the following, we
write $\nexts \fortau P$ to assert that it is not empty and to
directly denote the corresponding unique machine translation.

A consequence is that a machine process may only have flag names as
possible observable actions. In particular, it cannot have a coname
$\outn a$ as an observable action.

\begin{cor}
  \label{cor:wkobs}
  For all machine process $P$,
  $\WkObs P \subseteq \{ \flambda, \fenter, \fskip, \fdone, suc, z
  \}$.
\end{cor}

The second property we need is an operational correspondence between
the machine and its translation. 
\begin{lem}[Faithfulness]
  \label{lem:faithful}
  The following holds:
  \begin{itemize}
  \item $\sts \fnred{\fortau} \sts'$ iff
    $\nexts \fortau {\tr\sts} = \tr{\sts'}$;
  \item $\sts \fnred\fdone$ iff $\tr \sts \wtrans\tau \trans\fdone P$
    and $P \hobisim \nil$.
  \end{itemize}
\end{lem}

\begin{proof}
  The proof is by a straightforward case analysis on $C$ in the ``only
  if'' direction, and by case analysis on $\tr C$ in the ``if''
  direction.
\end{proof}

With these properties, we can prove that the translation from NFB
machines to \hocore processes is fully abstract. We define the process
complementing a flag as follows: for a flag $\flag \neq \fvr$, we
write~$\outn \flag$ for the process $\out \flag \nil$.  We define
$\outn \fvr$ inductively so that
$\outn{n + 1} \eqdef \out{suc}\nil \parallel \outn{n}$ and
$\outn{0} \eqdef \out{z}{\nil}$. The following lemma allows us to
reason up to $\tau$-transitions on the \hocore side.
\begin{lem}
  \label{lem:tau-implies-bisim}
  Let $P$, $P'$ be machine processes. If $P \wtrans\tau P'$ and
  $\WkObs P = \WkObs {P'}$, then $P \hobisim P'$.
\end{lem}

\begin{proof}
  One can check that
  $\{ (P \parallel R, P' \parallel R) \midd P \wtrans\tau P', \WkObs P
  = \WkObs{P'}\} \cup \{(R, R) \}$ is a barbed bisimulation using
  Lemma~\ref{lem:determinist}.
\end{proof}

\begin{lem}
  \label{lem:ho-lambda}
  If $\tr\sts \hobisim \tr{\sts'}$, then $\sts \ambisim \sts'$. 
\end{lem}

\begin{proof}
  We prove that
  \(\mathord{\relR} \eqdef \{ (\sts, \sts') \midd \tr\sts \hobisim
  \tr{\sts'} \} \) is a machine bisimulation. Let $\sts_1 \relR
  \sts'_1$.

  Suppose $\sts_1 \fnredrtc\tau \sts_2 \fnred \flag \sts_3$. By
  Lemma~\ref{lem:faithful}, we have
  \[
    \tr{\sts_1} \parallel \outn \flag \wtrans\tau \tr{\sts_2} \parallel
    \outn \flag \wtrans \tau \tr{\sts_3}.
  \]
  Because $\tr {\sts_1} \hobisim \tr{\sts'_1}$, there exists $P$ such
  that $\tr{\sts'_1} \parallel \outn \flag \wtrans\tau P$ and
  $\tr{\sts_3} \hobisim P$. Since $\WkObs{\tr{\sts_3}} = \WkObs P$, we
  have in particular $\neg{(P \obss{\outn \flag})}$ by
  Corollary~\ref{cor:wkobs}.  Consequently, there exists~$P'$ such that
  $\tr{\sts'_1} \wtrans\tau P'$, $P' \obss \flag$, and
  $P' \parallel \outn \flag \wtrans\tau P$. Consider
  $\nexts \flag {P'} = \tr{\sts'_3}$ for some~$\sts'_3$, which must exists since
  \(P'\) still has a barb on \(F\). We have
  $\tr{\sts'_1} \wtrans\tau P' \wtrans \flag \tr{\sts'_3}$, which
  implies $\sts'_1 \fnredrtc\tau \fnred \flag \sts'_3$ by
  Lemma~\ref{lem:faithful}. To show that $\sts_3 \relR \sts'_3$, we
  need to prove that $\tr{\sts_3} \hobisim \tr{\sts'_3}$. Because
  $\tr{\sts_3} \hobisim P$, it is enough to prove
  $P \hobisim \tr{\sts'_3}$.

  Suppose there is a choice process $Q$ so that
  $P' \parallel \outn \flag \wtrans \tau Q \wtrans\tau
  \tr{\sts'_3}$. It means that a communication on
  $\flag' \in \{ \fenter, \fskip \}$ happens, as it is the only
  possibility for a choice process to reach a translated
  configuration. The communication cannot be with $\outn \flag$,
  because that output is consumed by $P'$. Therefore, it implies that
  $P'\wkobs{\outn{\flag'}}$, but $P'$ is a machine process, so it is
  not possible by Corollary~\ref{cor:wkobs}.

  In the end, we have both $P' \parallel \outn \flag \wtrans\tau P$
  and $P' \parallel \outn \flag \wtrans\tau \tr{\sts'_3}$, and there
  is no choice process between $P' \parallel \outn \flag$ and
  $\tr{\sts'_3}$. As this is the only source of non-determinism, by
  Lemma~\ref{lem:determinist}, we either have $P \wtrans\tau
  \tr{\sts'_3}$ or $\tr{\sts'_3} \wtrans\tau P$. 

  We now prove that $\WkObs{\tr{\sts'_3}} = \WkObs{P}$, to be able to
  conclude with Lemma~\ref{lem:tau-implies-bisim}. The only way to have
  $\WkObs{\tr{\sts'_3}} \neq \WkObs{P}$ is if a choice has been made
  between $\tr{\sts'_3}$ and $P$, i.e., if there is a choice process $Q$
  making at least one $\tau$-transition between the two. The case
  $P \wtrans\tau Q \trans \tau \wtrans\tau \tr{\sts'_3}$ has been
  ruled out, because there is no choice process between
  $P' \parallel \outn \flag$ and $\tr{\sts'_3}$. The other possibility
  is $\tr{\sts'_3} \wtrans \tau Q \trans\tau \wtrans\tau P$, where \(Q\) is a
  choice process. Because a
  choice has been made, we have either
  $\WkObs P \subseteq \{ \fenter \}$ or
  $\WkObs P \subseteq \{ \fskip \}$. However, we have
  $\WkObs P = \WkObs {\tr{C_3}}$, and the translation of a
  configuration cannot have only a single flag of a non-deterministic
  choice as weak observable action, which leads to a
  contradiction. Consequently, we have
  $\WkObs P = \WkObs{ \tr{\sts'_3}}$, hence
  $P \hobisim \tr{\sts'_3}$ by Lemma~\ref{lem:tau-implies-bisim}.\\

  Suppose $\sts_1 \fnredrtc\tau \sts_2 \fnred \fdone$; then
  $\tr{\sts_1} \parallel \outn \fdone \wtrans\tau \tr{\sts_2}
  \parallel \outn \fdone \wtrans \tau P$ with $P \hobisim \nil$ by
  Lemma~\ref{lem:faithful}. Therefore, there exists $P'$ such that
  $\tr{\sts'_1} \parallel \outn \fdone \wtrans\tau P'$ and
  $P \hobisim P' \hobisim \nil$. With Lemma~\ref{lem:faithful}, we
  have $\sts'_1 \fnredrtc\tau \fnred\fdone$, as wished.
\end{proof}

\begin{lem}
  \label{lem:lambda-ho}
  If $\sts \ambisim \sts'$, then $\tr\sts \hobisim \tr{\sts'}$.
\end{lem}

\begin{proof}
  We prove that 
  \begin{align*}
    \relR & \eqdef  \left\{ 
    \begin{gathered}
      (P \parallel R, Q \parallel R) \midd \sts \ambisim \sts',
      \WkObs{P} = \WkObs{Q}, \fn R \cap \hidden = \emptyset \\
      P \wtrans\tau \tr{\sts}, Q \wtrans\tau \tr{\sts'} \mbox{ or
      } 
      \tr{\sts} \wtrans\tau P, \tr{\sts'} \wtrans\tau Q \\
    \end{gathered}
    \right\}\\
          & \cup \{ (P \parallel R, Q \parallel R) \midd P \hobisim \nil
            \hobisim Q, \fn R \cap \hidden = \emptyset \}
  \end{align*}
  is a barbed bisimulation. The congruence condition holds by
  definition of $\relR$, and the observable actions are the same
  because of the condition $\WkObs{P} = \WkObs{Q}$.  The
  reductions from $R$ are easy to match; we thus only check the
  reductions involving $P$ or~$Q$.
    
  If $P \trans\tau P'$, then we show that there exists $Q'$ such that
  $Q \parallel R \wtrans\tau Q' \parallel R$ and
  $P' \parallel R \relR Q' \parallel R$. If $P \wtrans\tau \tr\sts$
  and $Q \wtrans\tau \tr{\sts'}$, then taking $Q' = \tr{\sts'}$
  works. If $\tr\sts \wtrans\tau P$ and $\tr{\sts'} \wtrans\tau Q$,
  then we distinguish two cases. If $P$ is not a choice process, then
  $\WkObs{P} = \WkObs{P'}$, and we can simply choose $Q' =
  Q$. Otherwise, the transition $P \trans\tau P'$ is making a choice
  and $\WkObs{P'} \subsetneq \WkObs{P}$. But $\WkObs{P} = \WkObs{Q}$
  implies that $Q$ also reduces to a choice process. Then it is
  possible to choose the same branch as $P'$ and to define $Q'$
  accordingly.

  Suppose $P \parallel R \trans\tau P'$ with a communication on a
  flag $\flag$---a communication on another name is not possible
  because $\fn R \cap \hidden = \emptyset$. Then $R \obss{\outn
    \flag}$ and $P \obss \flag$, which is possible only in the case
  $\tr \sts \wtrans\tau P$. Therefore we also have $\tr{\sts'}
  \wtrans\tau Q$ by definition of $\relR$.

  Suppose $\flag \neq \fdone$: then we also have
  $P \trans \flag P'' \wtrans\tau \tr{\sts''}$ for some $P''$ with
  $\tr{\sts''} = \nexts \flag {P}$. By Lemma~\ref{lem:determinist},
  there exist $R'$ such that $P' = P'' \parallel R'$. We have
  $\tr\sts \wtrans \flag \tr{\sts''}$, which implies
  $\sts \fnredrtc\tau \fnred \flag \sts''$ by
  Lemma~\ref{lem:faithful}. Because $\sts \ambisim \sts'$, there
  exists~$\sts'''$ such that
  $\sts' \fnredrtc\tau \fnred \flag \sts'''$ and
  $\sts'' \ambisim \sts'''$. By Lemma~\ref{lem:faithful}, this implies
  $\tr{\sts'} \wtrans \flag \tr{\sts'''}$; in particular, there exists
  $Q''$ such that
  $\tr{\sts'} \wtrans\tau \trans \flag Q'' \wtrans\tau
  \tr{\sts'''}$. By Lemma~\ref{lem:determinist}, we also have
  $Q \wtrans\tau \trans \flag Q''$, therefore
  $Q \parallel R \wtrans\tau Q'' \parallel R'$. To show that
  $P'' \parallel R' \relR Q'' \parallel R'$, what is left to check is
  $\WkObs{P''} = \WkObs{Q''}$. We have $P'' \wtrans\tau \tr{\sts''}$,
  so there is no choice process between $P''$ and $\tr{\sts''}$,
  because a transition with label $\fenter$ or $\fskip$ is necessary
  to go from a choice process to the translation of a
  configuration. Therefore, we have $\WkObs{P''} =
  \WkObs{\tr{\sts''}}$, and similarly $\WkObs{Q''} =
  \WkObs{\tr{\sts'''}}$. Because $\sts'' \ambisim \sts'''$, any action
  from $\sts''$ is matched by $\sts'''$ and conversely, which implies
  $\WkObs{\tr{\sts''}} = \WkObs{\tr{\sts'''}}$ by
  Lemma~\ref{lem:faithful}. In the end, we have $\WkObs{P''} =
  \WkObs{Q''}$, and therefore $P'' \parallel R' \relR Q''
  \parallel R'$, as wished.

  If $\flag = \fdone$, then $P \trans\fdone P''$ for some $P''$ such
  that $P'' \hobisim \nil$. By Lemma~\ref{lem:determinist}, there
  exists~$R'$ such that $P'= P'' \parallel R'$. We have
  $\tr\sts \wtrans\tau \trans\fdone P''$, so
  $\sts \fnredrtc\tau \fnred\fdone$ holds by
  Lemma~\ref{lem:faithful}. Because $\sts \ambisim \sts'$, we have
  $\sts' \fnredrtc\tau \fnred \fdone$, so by Lemma~\ref{lem:faithful},
  there exists $Q''$ such that
  $\tr{\sts'} \wtrans\tau \trans\fdone Q''$ and $Q'' \hobisim
  \nil$. By Lemma~\ref{lem:determinist}, we also have
  $Q \wtrans\tau \trans \fdone Q''$, therefore we have
  $Q \parallel R \wtrans\tau Q'' \parallel R'$. The resulting
  processes $P'' \parallel R'$ and $Q'' \parallel R'$ are in the
  second set defining $\relR$.
\end{proof}

\begin{thm}
  \label{th:am-ho-nf}
  $\sts \ambisim \sts'$ iff $\tr\sts \hobisim \tr{\sts'}$.
\end{thm}

\begin{proof}
  By Lemmas~\ref{lem:ho-lambda} and~\ref{lem:lambda-ho}.
\end{proof}

As a result, we can deduce full abstraction between \hocore and the
$\lambda$-calculus with normal-form bisimilarity. 
\begin{cor}
  $t \nfbisim s$ iff there exists $n > \max{(\fv t \cup \fv s)}$ such that
  $\tr{\fneval t \mts n} \hobisim \tr{\fneval s \mts n}$.
\end{cor}

\begin{proof}
  By Theorems~\ref{th:nf-am} and~\ref{th:am-ho-nf}.
\end{proof}

The proofs of Lemmas~\ref{lem:ho-lambda} and~\ref{lem:lambda-ho} rely
on the translation being deterministic (Lemma~\ref{lem:determinist}) and
faithful (Lemma~\ref{lem:faithful}). In particular, flagged
non-determinism in machines is represented by choice processes
in \hocore, so that these processes make the choice observable with different
flags before reducing to the next translated configurations. The
machine translations defined in the upcoming sections satisfy these
properties, and the correspondence between machines and their
translation in \hocore can be proved in a similar fashion.

\section{Applicative Bisimilarity}
\label{sec:full-abstr-applicative}

Proving full abstraction \wrt normal-form bisimilarity requires
minimal interactions---synchronizations on flags---between a
machine process and the outside. Achieving full abstraction \wrt
applicative bisimilarity is intuitively more difficult, since this
bisimilarity tests $\lambda$-abstractions by applying them to an
arbitrary argument. Internalizing such bisimilarity is simple using
higher-order flags: one may think of the following transition
to test the result of a computation:
\begin{equation*}
\kam {\lam x t} \mts \fnred{s} \kam {\subst s x t} \mts
\end{equation*}
Although \hocore has higher-order communications, we cannot use them
to obtain a fully abstract encoding of such a machine for two
reasons. First, allowing interactions where the environment provides a
term may allow arbitrary processes to be received, including processes
that are not in the image of the translation, thus potentially
breaking invariants of the translation. Second, the translation of the
KAM has to hide the names it uses for the translation to be fully
abstract; it is thus impossible for the context to use such names and
to provide translated $\lambda$-terms to be tested.

We thus propose in this section to internalize applicative
bisimilarity using ordinary flags: when the abstract machine reaches a
value, it switches to a different mode where it non-deterministically
builds a test term step by step, using flags to indicate its choices
so as to ensure that a bisimilar machine builds the same term. The
translation of such a machine into \hocore is then similar to the
translation of the NFB machine.

Using simple flags to generate terms step by step implies we need to
deal with binders. In particular, and anticipating on the \hocore
translation, we no longer can rely on the definition of binding and
substitution from \hocore, as we cannot write a process that inputs a
translation of $t$ and outputs a translation of $\lambda x.t$ using an
\hocore binding for $x$. We thus switch to a pure data description of
bindings, using de Bruijn indices. As such terms still need to be
executed, we first recall the definition of the KAM with de Bruijn
indices and the definitions of contextual equivalence and applicative
bisimilarity for $\lambda$-terms with de Bruijn indices. We then
present the machine internalizing applicative bisimilarity, its
translation into \hocore, and show they are fully abstract. We finally
conclude this section by showing how contextual equivalence is
internalized in an abstract machine, generating contexts instead of
terms.

\subsection{The KAM and Behavioral Equivalences}
\label{ssec:KAM-ab}

In the $\lambda$-calculus with de Bruijn indices, a variable is a number that
indicates how many $\lambda$'s are between the variable and its binder. For
example, $\lam x x$ is written $\lamB 0$ and $\lam{xy}{x \app y}$ is written
$\lamB {\lamB {1 \app 0}}$. The syntax of terms ($t,s$), closures ($\clos$),
environments~($\envt,\envtt$), and values ($v$) is as follows.
\[
  t,s \bnf n \midd t \app s \midd \lamB t \qquad \clos \bnf (t, e) \qquad \envt,\envtt \bnf
  \clos \cons e \midd \mte \qquad v \bnf (\lamB t, e) \\  
\]
A \emph{closure} $\clos$ is a pair $(t, \envt)$ where $\envt$ is an
environment mapping the free variables of $t$ to closures;
environments are used in lieu of substitutions. A term $t$ is closed
if $\fv t = \emptyset$, and a closure $(t, \envt)$ is closed if the
number of elements of $\envt$ is bigger than the highest free variable
of~$t$, and $\envt$ is composed only of closed closures. 

The semantics is given by the original, environment-based KAM, where a
configuration~$C$ is now composed of the closed closure $(t, e)$ being
evaluated, and a stack $\stk$ of closures. The transitions rules are
as follows.
\[ 
C \bnf \kBeval {t} \envt {\stk} \mbox{ (configurations)} \qquad \stk \bnf \clos \cons
\stk \midd \mts \mbox{ (stacks)}
\]
\begin{align*}
\kBeval {t \app s} \envt \stk & \kBred{\tau}  \kBeval t \envt {(s,\envt) \cons \stk} \tag{\hname{Push}}\\
\kBeval 0 {(t, \envt) \cons \envtt} \stk & \kBred{\tau} \kBeval t \envt \stk
\tag{\hname{Zero}}\\
\kBeval {n+1}{(t, \envt) \cons \envtt} \stk & \kBred{\tau} \kBeval n \envtt \stk
\tag{\hname{Env}}\\
\kBeval {\lamB t} \envt {\clos \cons \stk} & \kBred{\tau}  \kBeval t {\clos \cons \envt}
\stk \tag{\hname{Grab}}
\end{align*}
In $\hname{Push}$, the argument $s$ of an application is
stored on the stack with its environment~$e$ while the term $t$ in function
position is evaluated. If we get a $\lambda$-abstraction $\lamB t$ (transition
$\hname{Grab}$), then an argument $\clos$ is moved from the stack to the top of
the environment to remember that $\clos$ corresponds to the de Bruijn index 0,
and the evaluation continues with~$t$. Looking up the closure corresponding to a
de Bruijn index in the environment is done with the rules $\hname{Env}$ and
$\hname{Zero}$. Because we evaluate closed closures only, it is not possible to
obtain a configuration of the form $\kBeval {n} \mte \stk$. If a configuration of
the form $\kBeval {\lamB t} \envt \mts$ is reached, then the evaluation is finished,
and the result is $(\lamB t, \envt)$.

\paragraph*{Behavioral equivalences} Contextual equivalence compares closed
terms by testing them within all contexts. A context $\ctx$ is a term with a
hole $\mtctx$ at a variable position; plugging a term $t$ in $\ctx$ is written
$\inctx \ctx t$. A context is closed if $\fv\ctx=\emptyset$. Contextual
equivalence is then defined as follows.
    
\begin{defi}
  Two closed terms $t$ and $s$ are contextually equivalent, written
  $t \ctxequiv s$, if for all closed contexts $\ctx$,
  \(\kBeval {\inctx \ctx t} \mte \mts \kBredrtc\tau \kBeval {\lamB
    {t'}} \envt \mts\) for some $t'$ and $\envt$ iff
  $\kBeval {\inctx \ctx s} \mte \mts \kBredrtc\tau \kBeval {\lamB
    {s'}} \envtt \mts$ for some $s'$ and $\envtt$.
\end{defi}
Contextual equivalence is characterized by applicative
bisimilarity~\cite{Abramsky-Ong:IaC93}, which reduces closed terms to values
that are then applied to an arbitrary argument.

\begin{defi}
  \label{def:app-cbn}
  A symmetric relation $\relR$ on closed closures is an applicative bisimulation
  if $(t, \envt) \relR (s, \envtt)$ and
  $\kBeval t \envt \mts \kBredrtc\tau \kBeval{\lamB {t'}}{\envt'} \mts$ implies that
  there exist $s'$ and $\envtt'$ such that
  $\kBeval s \envtt \mts \kBredrtc\tau \kBeval {\lamB {s'}}{\envtt'} \mts$, and for all
  closed $t''$, we have
  $(t', (t'', \mte) \cons \envt') \relR (s', (t'', \mte) \cons \envtt')$.
  
  Applicative bisimilarity $\appbisim$ is the largest applicative
  bisimulation.
\end{defi}
We can prove full abstraction between \hocore and the $\lambda$-calculus by
either internalizing contextual equivalence or applicative bisimilarity. We
choose the latter, as it is closer to normal-form bisimilarity. We briefly show
in Section~\ref{ssec:ctxequiv} how the machine of Section~\ref{ssec:app-machine}
can be adapted to handle contextual equivalence.

\subsection{Argument Generation for the Applicative Bisimilarity}
\label{ssec:app-machine}

After evaluating a term thanks to the KAM, we want to produce a closed
argument to pass it to the resulting value, and then restart the
evaluation process. The approach of \cite{Cai-Fu:MSCS11} consists in
waiting for a name from the outside, giving access to a translated
$\lambda$-term to be applied to. If the environment does not provide
access to a well-formed translation, the process simulating
$\beta$-reduction remains stuck. In contrast, our machine directly
generates a well-formed argument: we represent a $\lambda$-term as a
syntax tree, with de Bruijn indices at the leaves, and applications
and $\lambda$-abstractions at the nodes. We start generating from the
leftmost variable, and we then go left to right, meaning that in an
application, we create the term in function position before the
argument. These choices are completely arbitrary, as doing the
opposite---starting from the rightmost variable and go right to
left---is also possible. To be sure that we produce a valid, closed,
$\lambda$-term, we have to check that each de Bruijn index~$n$ has at
least $n+1$ $\lambda$-abstractions enclosing it, and that each
application node has two children.

To do so, we consider machine states with four components: the term $t$ being
constructed, a counter $\acc$ giving the minimal number of
$\lambda$-abstractions required to close the term, a stack~$\rho$ used to build
applications, whose syntax is
\[
  \rho \bnf (t, \acc) \cons \rho \midd \mtr \tag{application stacks}
\]
and which is explained in more detail later, and finally the closure $\clos$
for which the argument is being built. This last element is never modified by
the building process, and is just used to restart the machine in evaluation mode
when the argument is finished. We distinguish two kinds of states: the index
state $\kBint n \rho \clos$, where only de Bruijn indices can be built, and
the term state $\kBtm t \acc \rho \clos$, where any term can be produced. The
transitions for these states are given in Figure~\ref{fig:am-ab}, and we call
the extension of the KAM with these transitions the AB machine. 

\begin{figure}
\begin{align*}
\kBeval {\lamB t} \envt \mts & \kBred{\farg} \kBint 0 \mtr {(t, \envt)}
\tag{\hname{Arg}}\\
\\
\kBint n \rho \clos & \kBred{\fsuc} \kBint {n+1} \rho \clos
\tag{\hname{Suc}}\\
\kBint n \rho \clos & \kBred{\fvar} \kBtm  n {n+1} \rho \clos
\tag{\hname{Var}}\\
\kBtm t {\acc+1} \rho \clos & \kBred{\flam} \kBtm {\lamB t}{\acc} \rho \clos
\tag{\hname{Lambda}}\\
\kBtm t 0 \rho \clos & \kBred{\flam} \kBtm {\lamB t} 0 \rho \clos
\tag{\hname{Lambda0}}\\
\kBtm t {\acc} \rho \clos & \kBred{\fappfun} \kBint 0 {(t, \acc) \cons \rho} \clos
\tag{\hname{AppFun}}\\
\kBtm s {\acc_1} {(t, \acc_2) \cons \rho} \clos & \kBred{\fapp} \kBtm {t \app
  s}{\maxx {\acc_1}{\acc_2}} \rho \clos
\tag{\hname{App}}\\
\kBtm t 0  \mtr {(s, \envt)} & \kBred{\fdone} \kBeval s {(t, \mte) \cons \envt} \mts 
\tag{\hname{Restart}}
\end{align*}
\caption{AB machine: argument generation}
\label{fig:am-ab}
\end{figure}

The transition $\hname{Arg}$ starts the building process when we reach
a $\lambda$-abstraction in evaluation mode with the empty continuation
$\mts$. We begin with the index 0 and with the empty stack $\mtr$. The
value of the index can then be increased with the transition
$\hname{Suc}$. When we reach the needed value for the index, the
transition $\hname{Var}$ switches to the term mode, indicating that we
need at least $n+1$ $\lambda$-abstractions to close the term. We use two modes
to prevent a $\hname{Suc}$ transition on a term which is not an index.

In term mode, we can add $\lambda$-abstractions to the term,
decreasing $\acc$ if $\acc > 0$ with transition $\hname{Lambda}$, or
leaving $\acc$ at $0$ with transition $\hname{Lambda0}$; the
abstractions we introduce when $\acc=0$ do not bind any variable. Once
we are done building a term $t$ in function position of an
application, we use transition $\hname{AppFun}$ to build the argument
$s$. We start again in index mode, but we store on top of $\rho$ the
term $t$ with its counter $\acc_2$. When we finish $s$ with a counter
$\acc_1$, we build the application with transition $\hname{App}$,
which takes the maximum of $\acc_1$ and $\acc_2$ as the new minimal
number of $\lambda$-abstractions needed above $t \app s$. Note that
the $\hname{App}$ transition is allowed only if $\rho$ is not empty,
meaning that at least one $\hname{AppFun}$ has been done
before. Finally, we can conclude the term building process with
transition $\hname{Restart}$ only if $\acc=0$, meaning that all the
variables of the term are bound, and if $\rho$ is empty, meaning that
there is no application waiting to be finished.

\begin{exa}
  Figure~\ref{fig:ex-generation} presents how we generate the term
  $\lamB {\lamB {(\lamB {\underline 0}) \app (1 \app \lamB 0)}}$; we start with
  the underlined 0.
  \begin{figure}

\[    
\begin{array}{cccccccccc}
  \kBint 0 \mtr \clos
  & \kBred\fvar & \kBtm 0 1 \mtr \clos
  & \kBred\flam & \kBtm {\lamB 0} 0 \mtr \clos
  & \kBred\fappfun & \kBint 0 {(\lamB 0, 0) \cons \mtr} \clos
  \\
  & & \begin{tikzpicture}  
    \Tree [.$0$ ]
  \end{tikzpicture}
  & &
    \begin{tikzpicture}
      \tikzset{level distance=17pt}
      \Tree [.$\lambda$
      [.$0$ ] ]
    \end{tikzpicture}
  & &
      \begin{tikzpicture}
        \tikzset{level distance=17pt}
        \Tree [.{ } 
        [.$\lambda$ [.$0$
        ] ]
        \edge[draw=none];
        [.{ } \edge[draw=none]; $0$
        ] ]
      \end{tikzpicture} \\[4mm]
\end{array}
\]
\[
\begin{array}{cccccccccc}
  \kBred\fsuc \kBred\fvar
  & \kBtm 1 2 {(\lamB 0, 0) \cons \mtr} \clos
  & \kBred\fappfun & \kBint 0 {(1, 2) \cons (\lamB 0, 0) \cons \mtr} \clos
  \\
  & \begin{tikzpicture}
    \tikzset{level distance=17pt}
    \Tree [.{ } 
    [.$\lambda$ [.$0$
    ] ]
    \edge[draw=none];
    [.{ } \edge[draw=none]; $1$
    ] ]
  \end{tikzpicture}
  & &  \begin{tikzpicture}
    \tikzset{level distance=17pt}
    \Tree [.{ } 
    [.$\lambda$ [.$0$
    ] ]
    \edge[draw=none];
    [.{ } $1$
    \edge[draw=none];
    $0$ 
    ] ]
  \end{tikzpicture}\\[4mm]
\end{array}
\]
\[
  \begin{array}{cccccccc}
    \kBred\fvar\kBred\flam
    &\kBtm {\lamB 0} 0 {(1, 2) \cons (\lamB 0, 0) \cons \mtr} \clos
    & \kBred\fapp &\kBtm {1 \app {\lamB 0}} 2 {(\lamB 0, 0) \cons \mtr} \clos \\
    & \begin{tikzpicture}
      \tikzset{level distance=17pt}
      \Tree [.{ } 
      [.$\lambda$ [.$0$
      ] ]
      \edge[draw=none];
      [.{ } $1$
      \edge[draw=none];
      [.$\lambda$ $0$ ]
      ] ] 
    \end{tikzpicture}
    & &
   \begin{tikzpicture}
     \tikzset{level distance=17pt}
     \Tree [.{ }
     [.$\lambda$ [.$0$
     ] ]
     \edge[draw=none];
     [.{@} $1$
     [.$\lambda$ $0$ ]
     ] ] 
   \end{tikzpicture}\\[4mm]
  \end{array}
\]
\[
  \begin{array}{ccccccccc}
    \kBred\fapp
    & \kBtm {(\lamB 0) \app (1 \app {\lamB 0})} 2 \mtr \clos
    & \kBred\flam & \kBtm {\lamB {(\lamB 0) \app (1 \app {\lamB 0})}}
                    1 \mtr \clos
    \\
    & \begin{tikzpicture}
      \tikzset{level distance=17pt}
      \Tree [.{@} 
      [.$\lambda$ [.$0$
      ] ]
      [.{@} $1$
      [.$\lambda$ $0$ ]
      ] ] 
    \end{tikzpicture}
    & &
   \begin{tikzpicture}
     \tikzset{level distance=17pt}
     \Tree [.$\lambda$ [
     .{@} 
     [.$\lambda$ [.$0$
     ] ]
     [.{@} $1$
     [.$\lambda$ $0$ ]
     ] ] ]
   \end{tikzpicture}\\[4mm]
  \end{array}
\]
\[
  \begin{array}{ccccc}
    \kBred\flam& \kBtm {\lamB {\lamB {(\lamB 0) \app (1 \app {\lamB
                 0})}}} 0 \mtr \clos & \kBred \fdone  
    \\
  & \begin{tikzpicture}
    \tikzset{level distance=17pt}
    \Tree [.$\lambda$
    [.$\lambda$ [
    .{@} 
    [.$\lambda$ [.$0$
    ] ]
    [.{@} $1$
    [.$\lambda$ $0$ ]
    ] ] ] ]
  \end{tikzpicture}
  \end{array}
\]
\caption{Example of argument generation}
\label{fig:ex-generation}
\end{figure}

\end{exa}

Any closed term $t$ can be generated with the AB machine, and it is possible to
define the sequence of flags $\Seq t$ that will be raised. We
write $()$ for the empty sequence, and
$(\flag_1, \ldots, \flag_n, (\flag'_1, \ldots, \flag'_m), \flag_{n+1}, \ldots, \flag_l)$ for the
sequence $(\flag_1, \ldots, \flag_n, \flag'_1, \ldots, \flag'_m, \flag_{n+1}, \ldots, \flag_l)$.
\begin{defi}\label{def:seq-flag}
  Given a term $t$, we define $\Seq t$ as
  \begin{align*}
    \Seq t & \eqdef  (\Seqbis t, \fdone) &     \Seqbis n & \eqdef (\Seqint n, \fvar)     \\
    \Seqbis {t \app s} & \eqdef (\Seqbis t, \fappfun, \Seqbis s,
                         \fapp) & \Seqint 0 & \eqdef () \\
    \Seqbis {\lamB t} & \eqdef (\Seqbis t, \flam) &  \Seqint {n+1} & \eqdef  (\Seqint n, \fsuc) 
  \end{align*}

\end{defi}
We write $\sts \kBred{\Seq t} \sts'$ for
$\sts \kBred{\flag_1} \ldots \kBred{\flag_m} \sts'$ where
$\Seq t = (\flag_1, \ldots, \flag_m)$.
\begin{lem}
  \label{lem:abam-all-t}
  If $t'$ is closed, then
  $\kBint 0 \mtr {(t, \envt)} \kBred{\Seq {t'}} \kBeval t {(t', \mte) \cons \envt}
  \mts$.
\end{lem}

\begin{proof}
  Let $\clos = (t, e)$. First, for all $n$, we have
  $\kBint 0 \rho \clos {\kBred \fsuc}^n \kBred\fvar \kBtm n {n+1} \rho \clos$.
  Then we show by induction on $t'$ that
  $\kBint 0 \rho \clos  \;\;\kBred{\Seqbis {t'}} \;\; \kBtm {t'}{\acc} \rho \clos$ where
  $\acc = \max{(\fv{t'})}+1$ if $t'$ is not closed, and $\acc=0$ otherwise. The
  case $t = n$ is concluded with the previous observation. If
  $t' = t'_1 \app t'_2$, then
  \begin{align*}
    \kBint 0 \rho \clos  \;\;& \kBred{\Seqbis {t'_1}}  \;\;\kBtm {t'_1}{\acc_1} \rho \clos
                            \mbox { (by induction)} \\
                          & \kBred  \fappfun \;\;\kBint 0 {(t'_1, \acc_1) \cons
                            \rho} \clos \\
                          & \kBred{\Seqbis {t'_2}}  \;\;\kBtm {t'_2}{\acc_2} {(t'_1, \acc_1) \cons
                            \rho} \clos \mbox { (by induction)} \\
                          & \kBred  \;\;\fapp \kBtm {t'_1 \app t'_2}{\max{(\acc_1,
                            \acc_2)}} \rho \clos 
  \end{align*}
  By case analysis on $(\acc_1, \acc_2)$, one can check that
  $\max{(\acc_1, \acc_2)}$ is the desired value. If $t' = \lamB {t''}$, then by
  induction, we have
  $\kBint 0 \rho \clos \;\;\kBred{\Seqbis {t''}} \;\;\kBtm {t''}{\acc} \rho \clos$,
  and then
  $ \kBtm {t''}{\acc} \rho \clos \kBred\flam \kBtm {t'}{\acc'} \rho \clos$ where
  $\acc'$ is as wished depending on~$\acc$.

  This implies that for a closed term $t'$, we have
  $\kBint 0 \mtr \clos \;\;\kBred{\Seqbis {t'}} \;\;\kBtm {t'} 0 \mtr \clos$, and the
  last transition $\kBred\fdone$ gives what we want. 
\end{proof}

This lemma allows us to prove the correspondence between the AB machine and
applicative bisimilarity. 

\begin{thm}
  $(t, \envt) \appbisim (s, \envtt)$ iff $\kBeval t \envt \mts
  \ambisim \kBeval s \envtt \mts$.
\end{thm}

\begin{proof}
  To prove that machine bisimilarity implies applicative bisimilarity, we show
  that
  \[
    \mathord{\relR} \eqdef \{ ((t, \envt), (s, d)) \midd \kBeval t e \mts \ambisim
    \kBeval s \envtt \mts \}
  \]
  is an applicative bisimulation. Suppose
  $\kBeval t \envt \mts \kBredrtc\tau \kBeval {\lamB {t'}}{\envt'} \mts$. Because
  we have also
  $\kBeval {\lamB {t'}}{\envt'} \mts \kBred\farg \kBint 0 \mtr {(t', \envt')}$, there
  exist $s'$ and $\envtt'$ such that
  \[\kBeval s \envtt \mts \kBredrtc\tau \kBeval {\lamB {s'}}{\envtt'} \mts \kBred\farg
    \kBint 0 \mtr {(s', \envtt')}\] and
  $\kBint 0 \mtr {(t', \envt')} \ambisim \kBint 0 \mtr {(s',
    \envtt')}$. Let $t''$ be a closed term. By
  Lemma~\ref{lem:abam-all-t}, we have
  $\kBint 0 \mtr {(t', \envt')} \kBred{\Seq {t''}} \kBeval {t'}{(t'',
    \mte) \cons \envt'} \mts$, which can only be matched by
  $\kBint 0 \mtr {(s', \envtt')} \kBred{\Seq {t''}} \kBeval {s'}{(t'',
    \mte) \cons \envtt'} \mts$. From
  $\kBint 0 \mtr {(t', \envt')} \ambisim \kBint 0 \mtr {(s',
    \envtt')}$, we get
  $\kBeval {t'}{(t'', \mte) \cons \envt'} \mts \ambisim \kBeval
  {s'}{(t'', \mte) \cons \envtt'} \mts$. It implies that
  $(t', (t'', \mte) \cons \envt') \relR (s', (t'', \mte) \cons
  \envtt')$ holds, as wished.

  For the reverse implication, we show that
  \begin{align*}
    \mathord{\relR} & \eqdef \{ (\kBeval t \envt \mts, \kBeval s \envtt \mts) \midd (t, \envt)
                      \appbisim (s, \envtt) \} \\
                    & \cup \{ (\kBint n \rho {(t, e)}), \kBint n \rho
                      {(s, d)}) \midd (\lamB t, e) \appbisim (\lamB s, d) \}\\
                    & \cup \{ (\kBtm {t'} \acc \rho {(t, e)}), \kBtm {t'} \acc \rho
                      {(s, d)}) \midd (\lamB t, e) \appbisim (\lamB s, d) \}
  \end{align*}
  is a machine bisimulation. Suppose
  $\kBeval t \envt \mts \relR \kBeval s \envtt \mts$. The only
  possible flagged transition is
  $\kBeval t \envt \mts \mathrel{\kBredrtc\tau \kBred\farg} \kBint 0 \mtr
  {(t', e')}$, which implies in particular
  $\kBeval t \envt \mts \kBredrtc\tau \kBeval {\lamB {t'}}{e'} \mts$.
  Because $(t, \envt) \appbisim (s, \envtt)$, there exists
  $(s', \envtt')$ such that
  $\kBeval s d \mts \kBredrtc\tau \kBeval {\lamB {s'}}{d'} \mts$ and
  $(\lamB {t'}, e') \appbisim (\lamB{s'}, d')$. We also have
  $\kBeval {\lamB {s'}}{\envtt'} \mts \kBred\farg \kBint 0 \mtr
  {(s', \envtt')}$, hence
  $\kBeval s \envtt \mts \mathrel{\kBredrtc\tau \kBred\farg} \kBint 0 \mtr
  {(s', \envtt')}$ and
  $\kBint 0 \mtr {(t', e')} \relR \kBint 0 \mtr {(s', \envtt')}$
  hold, as wished.

  To cover the last two sets of the bisimulation, let
  $\kBarg {t'} \acc \rho {(t, \envt)} \relR \kBarg {t'} \acc \rho {(s,
    \envtt)}$ with $\mathsf{arg} \in \{\mathsf{ind}, \mathsf{tm} \}$.
  If $\kBarg {t'} \acc \rho {(t, e)} \kBred \flag \sts$ with
  $\flag \neq \fdone$, then
  $\kBarg {t'} \acc \rho {(s, \envtt)} \kBred \flag \sts'$ where~$\sts'$
  is the same as $\sts$ except for the closures. These are not changed
  by the transition and are still applicative bisimilar; we have
  therefore $\sts \relR \sts'$.  Otherwise, we have
  $\kBtm {t'} \acc \rho {(t, \envt)} \kBred \fdone \kBeval t {(t',
    \mte) \cons e} \mts$, and also
  $\kBtm {t'} \acc \rho {(s, \envtt)} \kBred \fdone \kBeval s {(t',
    \mte) \cons \envtt} \mts$. The resulting configurations are in
  $\relR$ because $(\lamB t, \envt) \appbisim (\lamB s, \envtt)$
  implies
  $(t, (t', \mte) \cons \envt) \appbisim (s, (t', \mte) \cons \envtt)$
  for all closed $t'$.
\end{proof}

\subsection{Translation into \hocore}
\label{ssec:trans-AM-AB}

We detail each component of the translation of the AB machine into
\hocore, starting with the evaluation mode, i.e., the KAM.
\begin{align*}
  \tr{\kBeval t e \stk} 
  & \eqdef \tr t \parallel \out{\env}{\tr e} \parallel \out c
    {\tr \stk} \parallel \Prec
  & \tr {(t, e)}
  & \eqdef \out{\clos_1}{\tr t} \parallel \out{\clos_2}{\tr e} \\ 
  \tr {\clos \cons \stk}
  & \eqdef \out \hp {\tr \clos}
    \parallel \out c {\tr \stk}
  & \tr {\clos \cons e}
  & \eqdef \out \he {\tr \clos} \parallel \out {\env}{\tr e} \\
  \tr \mts & 
             \eqdef \ldots
  & \tr \mte & \eqdef \nil
\end{align*}

We follow the same principles as in Section~\ref{sec:KAM}: a non-empty
stack~$\stk$ or environment~$e$ is represented by a pair of messages,
respectively on~$\hp$ and~$c$, and $\he$ and $\env$ (we use the longer \(\env\)
for legibility reasons). A closure is
represented by two messages, one containing the term on $\clos_1$ and
the other the environment on $\clos_2$. The process representing
the empty environment~$\mte$ should never be executed, because all the
closures we manipulate are closed; as a result, we can choose any
process to represent it, e.g.,~$\nil$. The empty stack $\tr \mts$ and
the process $\Prec$ are used to generate an argument and are defined
and explained later.
\begin{align*}
  \tr {t \app s} & \eqdef \appeval {\tr t}{\tr s} \\
  \appeval{P_t}{P_s} &\eqdef \inp c x {\inp {\env} y {(P_t \parallel \out c {\out \hp
                         {\out {\eta_1}{P_s} \parallel \out {\eta_2}{y}} \parallel
                         x} \parallel \out{env} y )}}\\
  \tr {\lamB t} & \eqdef \lameval{\tr t} \\
  \lameval{P_t} &\eqdef \inp c x {(x \parallel
                  \inp \hp y {\inp {env} z 
                  {(P_t \parallel \out
                  {env}{\out \he {y} \parallel
                  \out{env} z})}} )}
  \\
  \tr {n+1} & \eqdef \intevalf{\tr n} \\
  \intevalf{P_n} &\eqdef \inp{env} x {(x \parallel \inp \he \osef {P_n})}\\
  \tr 0 & \eqdef \inp \env x {(x \parallel \inp \he y {\inp \env \osef
          {(y \parallel \inp {\clos_1}{y_1}{\inp {\clos_2}{y_2}{(y_1 \parallel
          \out \env {y_2})}})}})} 
\end{align*}

The encoding of $t \app s$ simulates the rule $\hname{Push}$: we receive the
current stack and environment $e$ to create the new stack with $(s, e)$ on
top. Because we receive the current environment to put it on the stack, we have
to recreate it on $\env$, unchanged. In the encoding of $\lamB t$, we capture
the stack and environment, and if the stack is non-empty, we fetch its head
$\clos$ to create a new environment with $\clos$ on top. Finally, a de Bruijn
index $n>0$ goes through the current environment, until we reach the correct
closure (case $n=0$). In that case, we receive the head $\clos$ and tail of the
environment, we discard the tail as it is no longer useful, and we restore the term
and environment stored in $\clos$.\\

If $\lamB t$ is run in the environment $e$ and the empty stack $\mts$, then we
obtain
\( \tr \mts \parallel \stuck {(t,e)} \),
where
\(
\stuck {(t,e)} \eqdef \inp \hp z {(\tr t \parallel \out \env {\out \he z \parallel
    \out \env {\tr e}})}.
\)
The process $\tr \mts$ has to start the argument generating process, and the result has
then to be sent on $\hp$ for the evaluation to restart. The process $\stuck {(t,e)}$ remains
stuck during the whole generation process. We now explain how $\kBint n
\rho \clos$ and $\kBtm t \acc \rho \clos$ are encoded, starting with $\acc$ and
$\rho$.
\begin{align*}
  \tracc{0}
  &\eqdef\inp{zero} x {\inp {suc} \osef x}
  & \tr \mtr
  &\eqdef \inp{mt} x {\inp {cs} \osef x} \\ 
  \tracc{\acc+1}
  & \eqdef \suk{\tracc \acc}
  & \tr {(t, \acc) \cons \rho}
  &\eqdef  \ConsR {\tr{(t, \acc)}}{\tr \rho} \\
  \suk{P_\acc}
  & \eqdef \out {suk}{P_\acc} \parallel \inp{zero} \osef {\inp {suc} x
    x}
  & \ConsR {P_{hd}}{P_\rho}
  & \eqdef \out \hr{P_{hd}} \parallel \out r {P_\rho} \parallel
    \inp{mt} \osef {\inp {cs} x x}\\
  & & \tr {(t, \acc)}
  & \eqdef \out {w_1}{\trBeval t} \parallel \out {w_2}{\tracc {\acc}}
\end{align*}

The machine distinguishes cases based on whether $\acc$ is $0$ or not,
to know if we should apply the transition $\hname{Lambda}$ or
$\hname{Lambda0}$. In the encoding of these rules (see the definition
of $\rawLambdaa$ below), we send on name $zero$ the expected behavior
if $\acc=0$, and on $suc$ what to do otherwise. The translation of the
counter receives both messages, executes the corresponding one (e.g.,
the one on $zero$ for the encoding of $0$), and discards the
other. Apart from that, $\acc$ is translated as nested sending on name \(suk\).
Similarly, the translation of $\rho$ combines the regular
encodings of pairs and stacks, but also indicates whether $\rho$ is
empty or not, to know if we can apply the transitions $\hname{App}$
and $\hname{Restart}$.
\begin{align*}
  \tr \mts & \eqdef \inp{\farg} \osef {(\out{ind}{\trBeval 0} \parallel
             \out{k}{\tracc 1} \parallel \out
             r {\tr \mtr} \parallel \out{initInd} \nil)}\\[5mm]
  \tr{\kBint n \rho \clos} & \eqdef \out {ind} {\trBeval n} \parallel \out
                                    {k}{\tracc {n+1}} \parallel \out r {\tr \rho} \parallel 
                                    \Prec \parallel \out{initInd} \nil \parallel
                                    \stuck \clos\\
  \RecInt & \eqdef \inp {initInd} \osef {\inp {recind} x {(x \parallel \out {recind} x
            \parallel \Succ + \Var)}}\\
  \Succ & \eqdef \inp{\fsuc} \osef {\inp {ind} x {\inp {k} y {(\out {ind}{\intevalf
          x} \parallel \out{k}{\suk y} \parallel \out {initInd} \nil)}}} \\
  \Var & \eqdef \inp{\fvar} \osef {\inp {ind} x {(\out {tm} x \parallel \out {initTm}
         \nil)}} \\[5mm]
  \tr{\kBtm t {\acc} \rho \clos} & \eqdef \out {tm} {\trBeval t} \parallel \out
                                   {k}{\tracc {\acc}} \parallel \out r {\tr
                                   \rho} \parallel \Prec \parallel \out{initTm}
                                   \nil \parallel 
                                   \stuck \clos
\end{align*}

After flagging $\farg$, the process $\tr\mts$ starts the argument generation
process in index mode: the index being built is sent on $ind$ (here, initialized
with $\tr 0$) and the stack on $r$. We also build on $k$ the
counter $\acc$ corresponding to the index: at any point, we have $\acc
= n+1$. We keep two messages on $ind$ and $k$ encoding almost the same
information, but with different encodings, as the index $n$ and the
counter $\acc$ are used differently.

The message on $initInd$ triggers the recursive process $\RecInt$,
which non-deterministically chooses between $\Succ$ and
$\Var$. Executing $\Succ$ flags $\fsuc$, increases the values of the
index (thanks to $\rawinteval$) and the counter (with $\rawsuk$), and
relaunches the $\RecInt$ process with a message on
$initInd$. Executing $\Var$ flags $\fvar$, moves the index from $ind$
to $tm$, and initiates the term mode by sending a message on $initTm$,
which triggers the recursive process $\RecTm$.
\begin{align*}
  \RecTm & \eqdef initTm(\osef).rectm(x).\\
         & \hspace{1em}\left(
           \begin{aligned}
             x \parallel \out {rectm} x \parallel \inp r y {
               y &\parallel \overline{cs}\left\langle
                 \begin{aligned}
                   &\inp k z {\inp \hr {y_1}{\inp r {y_2}{\\
                         & (\Lambdaa z \parallel \out r y)+ \AppFun z y+ \App{z}{y_1}{y_2}}}}
                 \end{aligned}
               \right\rangle \\
               &\parallel \overline{mt} \left\langle \inp k z {\left(
                     \begin{aligned}
                       z &\parallel
                       \overline{zero}\left\langle
                         \begin{aligned}
                           & \quad(\Lambdaa z \parallel \out r {\tr \mtr}) \\
                           & + \AppFun z {\tr \mtr}+ \Done 
                         \end{aligned}
                       \right\rangle
                       \\
                       &\parallel \outn{suc}\left\langle\inp{suk} \osef {
                           \begin{aligned}
                             &\quad(\Lambdaa{z} \parallel \out r {\tr\mtr})
                             \\&+ \AppFun z {\tr\mtr}
                           \end{aligned}
                         \right\rangle}
                     \end{aligned}
                   \right)}
               \right\rangle}
             \end{aligned}
  \right)
\end{align*}

The goal of $\RecTm$ is to non-deterministically choose between the
four transitions available in term mode, namely $\kBred\flam$,
$\kBred\fappfun$, $\kBred\fapp$, and $\kBred\fdone$. However, some of
these transitions have requirements: $\kBred\fapp$ needs
$\rho \neq \mtr$ and $\kBred\fdone$ needs $\rho = \mtr$ and
$\acc = 0$. The process $\RecTm$ is therefore doing a case analysis to
check these conditions. First, it captures $\tr\rho$ on $r$: if
$\rho \neq \mtr$, it executes the message on $cs$, which makes a
choice between $\flam$, $\fappfun$, and~$\fapp$, which are represented
by respectively $\rawLambdaa$, $\rawAppFun$, and $\rawApp$. If
$\rho=\mtr$, then we do a case analysis on $\acc$. If $\acc=0$, then
we can do either $\flam$, $\fappfun$, or $\fdone$ (represented by
$\Done$), otherwise, only $\flam$ or~$\fappfun$ are possible.
\begin{align*}
  \Lambdaa{P_\acc}
  & \eqdef \inp{\flam} \osef {\inp{tm} x
    {\left(\begin{aligned}
          \out {tm}{\lameval x} \parallel P_\acc
          & \parallel \out{zero}{\out {k}{P_\acc} \parallel \out{initTm} \nil} \\
          & \parallel  \out{suc}{\inp {suk} y {(\out{k} y \parallel
              \out{initTm} \nil)}}
        \end{aligned}\right)}} \\
  \AppFun{P_\acc}{P_\rho}
  & \eqdef \inp \fappfun \osef {\inp {tm} x}{
    \left(
    \begin{aligned}
      & \out r {\ConsR  {\out \hr {\out {w_1}{x} \parallel
            \out{w_2}{P_\acc}}}{P_\rho}} \parallel \\
      & \out{ind}{\trBeval 0} \parallel \out{k}{\tracc
        1} \parallel \out{initInd} \nil)
    \end{aligned}
        \right)} \\
  \App{P_\acc}{P_{hd}}{P_\rho}
  & \eqdef \fapp(\osef).tm(x_2).\\
  & \hspace{1em} \left(P_{hd} \parallel 
    \inp{w_2}{y}{\inp{w_1}{x_1}{\left(
    \begin{aligned}
      &\out{max1}{y} \parallel \out{max2}{P_\acc} \parallel
      \out{init1}{y} \parallel \out{init2}{P_\acc} \parallel \\
      &\inp{resu}{z}{
        \left(
          \begin{aligned}
            &\out{tm}{\appeval{x_1}{x_2}} \parallel
            \out r {P_\rho} \parallel \\
            &\out{k}{z} \parallel \out{initTm} \nil
        \end{aligned}
        \right)}
    \end{aligned}
        \right)}}
        \right)
\end{align*}
The process $\rawLambdaa$ adds a $\lambda$-abstraction to the term
in~$tm$, updating $\acc$ (represented by~$P_\acc$) accordingly: if
$\acc=0$, then it is restored unchanged on~$k$, otherwise, it is
decreased by 1 by releasing the message in $suk$. The process
$\rawAppFun$ pushes on the stack~$P_\rho$ the current term~$t_1$ on
$tm$ and its counter $\acc_1$ (represented by $P_\acc$), which is the
term in function position of an application. It then relaunches the
index mode to build the argument $t_2$ with its counter~$\acc_2$. The
process $\rawApp$ can then build the application itself, by computing
the maximum between~$\acc_1$ and~$\acc_2$ with the processes $\RecMax$
and $\rawMax$.
\begin{align*}
  \RecMax & \eqdef \inp{init1}{x_1}{\inp{init2}{x_2}{\inp{recmax}{y}{(y \parallel
            \out{recmax} y \parallel \Max {x_1}{x_2})}}} \\
  \Max{P_1}{P_2} & \eqdef P_1
                   \begin{aligned}[t]
                     & \parallel \out{zero}{\inp{max2} x {\out{resu}
                         x}} \\
                     & \parallel \outn{suc}\left\langle\inp{suk}{x_1}{\left( 
                           \begin{aligned}
                             P_2 &\parallel \out{zero}{\inp{max1} x {\out{resu} x}} \\
                             & \parallel
                             \out{suc}{\inp{suk}{x_2}{(\out{init1}{x_1} \parallel \out{init2}{x_2})}} 
                           \end{aligned}
                         \right)}\right\rangle
                   \end{aligned}
\end{align*}
We compute the maximum between $\acc_1$ and $\acc_2$ by removing the layers of
successors common to $\acc_1$ and $\acc_2$, until we reach~$0$ for one of
them. If we reach~$0$ for $\acc_1$ first, then $\acc_2$ is the max, otherwise it
is $\acc_1$. We store the initial values of $\acc_1$ and $\acc_2$ in
respectively $max1$ and $max2$, and the decomposition occurs in $\rawMax$, where
$init1$ is initialized with~$\acc_1$ and $init2$ with $\acc_2$. If
$P_1 = \tracc{\acc_1} = \tracc 0$, then we send $\acc_2$ (stored in $max2$) on
$resu$. Otherwise, $P_1 = \out{suk}{P'_1}$, and we do a case analysis on the
process $P_2=\tracc{\acc_2}$. If $P_2=\tracc 0$, then we send $\acc_1$ on
$resu$, otherwise $P_2 = \out{suk}{P'_2}$, and we restart $\RecMax$ by
sending~$P'_1$ and $P'_2$ on $init1$ and $init2$, respectively. Once the max is
known on $resu$, then $\rawApp$ builds the application $t_1 \app t_2$ and
relaunches $\RecTm$.

\begin{align*}
  \Done & \eqdef \inp \fdone \osef {\inp {tm} x {(\out \hp {\out
          {\clos_1} x \parallel \out {\clos_2}{\tr \mte}} \parallel \out c
          {\tr \mts})}}\\
  \Prec & \eqdef \RecInt \parallel \out{recind}{\RecInt} \parallel \RecTm \parallel
          \out{rectm}{\RecTm} \parallel \RecMax \parallel \out{recmax}{\RecMax}\\[-3mm]
\end{align*}

Finally, the process $\Done$ ends the argument generation phase, and restarts
the computation by restoring the empty continuation and by passing the term in
$tm$ to $\stuck \clos$. The process $\Prec$ contains all the processes necessary
to encode the different recursive mechanisms.

\paragraph*{Full abstraction} A case analysis on the configurations of
the machine shows that the translation of the AB machine satisfies the
conditions of Section~\ref{ssec:fullabs-nf}, and we can prove a full
abstraction result between AB machines and \hocore similar to
Theorem~\ref{th:am-ho-nf}. We deduce a full
abstraction result between the $\lambda$-calculus with applicative
bisimilarity and \hocore.

\begin{cor}
  \label{cor:fullabs-ab}
  For all closed closures $(t, \envt)$ and $(s, \envtt)$,
  $(t, \envt) \appbisim (s, \envtt)$ iff
  $\tr{\kBeval t \envt \mts} \hobisim \tr{\kBeval s \envtt \mts}$.
\end{cor}

\subsection{Internalizing Contextual Equivalence}
\label{ssec:ctxequiv}

Corollary~\ref{cor:fullabs-ab} is enough to deduce full abstraction
w.r.t. contextual equivalence, since
$t \ctxequiv s \iff (t,\mte) \appbisim (s,\mte)$. However, it is
possible to prove a similar result directly, by internalizing
contextual equivalence in an abstract machine. We do so to show how
the resulting contextual equivalence machine differs from the AB
machine. 

\begin{figure}
\begin{align*}
\ctxind n {\acc} \rho & \kBred{\fsuc} \ctxind {n+1}{\acc+1} \rho 
\tag{\hname{Suc}}\\
\ctxind n {\acc} \rho & \kBred{\fvar} \ctxtm  n {\acc} \rho 
\tag{\hname{Var}}\\
\ctxtm t {\acc+1} \rho  & \kBred{\flam} \ctxtm {\lamB t}{\acc} \rho 
\tag{\hname{Lambda}}\\
\ctxtm t 0 \rho  & \kBred{\flam} \ctxtm {\lamB t} 0 \rho 
\tag{\hname{Lambda0}}\\
\ctxtm t {\acc} \rho  & \kBred{\fpush} \ctxind 0 1 {(t, \acc) \cons \rho} 
\tag{\hname{AppPush}}\\
\ctxtm t {\acc_1} {(s, \acc_2) \cons \rho}  & \kBred{\fappl} \ctxtm {t \app
  s}{\maxx {\acc_1}{\acc_2}} \rho 
\tag{\overleftarrow{\hname{App}}}\\
\ctxtm t {\acc_1} {(s, \acc_2) \cons \rho}  & \kBred{\fappr} \ctxtm {s \app
  t}{\maxx {\acc_1}{\acc_2}} \rho 
\tag{\overrightarrow{\hname{App}}}\\
\ctxtm t 0  \mtr  & \kBred{\fstart} \kBeval t \mte \mts 
\tag{\hname{Start}} \\
\\
\kBeval {\lamB t} \envt \mts & \kBred\fdone \tag{\hname{Done}}
\end{align*}
\caption{Contextual equivalence machine}
\label{fig:am-ctx}
\end{figure}

Figure~\ref{fig:am-ctx} gives the transitions of the contextual
equivalence machine, except for the $\kBred\tau$ transitions, which
are the same as in Section~\ref{ssec:KAM-ab}. In contrast with the AB
machine, the contextual equivalence machine produces a context first,
and then reduces the resulting term; consequently, the starting point
is a state $\ctxtm t 0 \mtr$, where $t$ is the closed term we want to
plug in the context. When the context is finished, the transition
$\kBred\fstart$ switches to the evaluation mode. Also, the evaluation
part of the machine is not executed several times, since $\ctxequiv$
is not coinductive. We flag $\fdone$ when the evaluation terminates,
to distinguish a terminating term from a diverging one.

Creating a context $\ctx$ is almost the same as generating an argument
in the AB machine, except that we want to plug a closed term $t$
inside. We build $\inctx \ctx t$ by starting the generation process
from $t$; $t$ can be anywhere in $\inctx \ctx t$, not necessarily at
the leftmost position, so we cannot do the generation process going
left to right in an application, as with the AB machine
(Section~\ref{ssec:app-machine}). Instead, after producing a term $t$
and with a term $s$ on the stack, we can do either the transition
$\hname{AppLeft}$ to build $t \app s$, or $\hname{AppRight}$ to build
$s \app t$.

\begin{exa}
  We show how to generate the context $\lamB {(0 \app 0) \app (\mtctx \app 0)}$
  around $t$.
  \begin{align*}
    \ctxtm t 0 \mtr & \kBred\fpush \kBred\fvar \ctxtm 0 1 {(t, 0) \cons \mtr} \\
                    & \kBred\fappr \ctxtm {t \app 0} 1 \mtr \\
                    & \kBred\fpush\kBred\fvar \ctxtm 0 1 {(t \app 0, 1) \cons \mtr} \\
                    & \kBred\fpush \kBred\fvar \ctxtm 0 1 {(0, 1) \cons (t \app 0, 1) \cons \mtr} \\
                    & \kBred\fappr \ctxtm {0 \app 0} 1 {(t \app 0, 1) \cons \mtr} \\
                    & \kBred\fappl \ctxtm {(0 \app 0) \app (t \app 0)} 1 \mtr \\
                    & \kBred\flam \ctxtm {\lamB {(0 \app 0) \app (t \app 0)}} 0 \mtr
  \end{align*}
\end{exa}

The translation of the contextual equivalence machine into \hocore and the full
abstraction proofs are similar to the AB machine ones.
\begin{thm}
  \label{cor:fullabs-ab-thm}
  If $t$ and $s$ are closed terms, then $t \ctxequiv s$ iff
  $\tr{\ctxtm t 0 \mtr} \hobisim \tr{\ctxtm s 0 \mtr}$.
\end{thm}

\section{Call-by-Value}
\label{sec:cbv}

We adapt our techniques to the call-by-value setting, by internalizing
equivalences into the CK machine.

\subsection{The CK machine} As in call-by-name, we internalize
normal-form bisimilarity into a substitution-based abstract machine,
namely the CK machine~\cite{Felleisen-al:09}. 
\begin{align*}
  C & \bnf \kam{t}{\stk} \midd \kamC \stk v  \tag{configurations}\\
  v & \bnf x \midd \lam x t \tag{values} \\
  t,s &\bnf  v \midd t \app s \tag{terms}\\
  \stk &\bnf \starg t \cons \stk \midd \stfun v \cons \stk \midd \mts\tag{stacks} \\
  \\[-3mm]
  \kam{t \app s} \stk & \kred \kam t {\starg s \cons \stk} \tag{\hname{Fun}}\\
  \kam v \stk & \kred \kamC \stk v \tag{\hname{Switch}}\\
  \kamC {\starg t \cons \stk} v & \kred \kam t {\stfun v \cons \stk} \tag{\hname{Arg}}\\
  \kamC {\stfun {\lam x t} \cons \stk} v  & \kred \kam {\subst v x t}
                                            \stk \tag{\hname{Beta}} 
\end{align*}
The machine distinguishes configurations of the form $\kam t \stk$,
which evaluates $t$ in the stack~$\stk$, from configurations
$\kamC \stk v$, which decide how computation should proceed depending
on the stack $\stk$. In left-to-right call-by-value evaluation, we
evaluate the function before reducing its argument. What we remember
on the stack reflects this order: $\starg t$ means that we evaluate
the function and remembers its argument $t$, while $\stfun v$ means
that we evaluate the argument of an already computed function $v$.

The $\hname{Fun}$ machine step focuses on the term in function
position. When it is evaluated, we look at the stack by switching to
the continuation mode (rule $\hname{Switch}$). Either the top of the
stack is an argument which needs to be evaluated (rule $\hname{Arg}$),
or it is a $\lambda$-abstraction ready to be applied to its argument
(rule $\hname{Beta}$). When we reduce closed terms, we cannot get a
configuration of the form $\kamC {\stfun x \cons \stk} v$.

The encoding of the CK machine in \hocore is given in
Figure~\ref{fig:ck-hocore}. Unlike in the KAM, the stack has some
control on how the reduction proceeds. It is reflected in the encoding
of $\kamC \stk v$, where the encoded value $\trv v$ is sent on a
channel name $\nval$, waiting to be consumed by the stack. The top of
the stack decides what becomes of $\trv v$: in the case of $\starg t$,
the value $\trv v$ is forwarded to the top of the stack and $\tr t$
can reduce. For $\stfun{\lam x t}$, the encoded $\lambda$-abstraction
$\inp \nval x {\tr t}$ is receiving $\trv v$, and the computation
continues as $\subst {\trv v} x {\tr t}$: as for the KAM, we encode
$\beta$-reduction by a \hocore communication.

The encoding of values considered as terms $\tr v$ is uniform, and
recreates the encoding of $\kamC \stk v$ by capturing the current
stack and running it in parallel with $\out \nval {\trv v}$. The
encoding of the application $\tr {t \app s}$ is as expected: it runs
$\tr t$ and pushes $\tr s$ on the current stack. Like for the KAM, the
encoding of the CK machine needs only two names, $c$ and $\nval$.

We have a one-to-one correspondence between the CK-machine and its
translation, a result stronger than for the KAM
(cf. Theorem~\ref{th:cor-KAM-trans}): the administrative
$\hname{Switch}$ step, which does not exist in the KAM, corresponds to
a communication in the \hocore translation.

\begin{thm}
  If $C \kred C'$, then $\tr C \trans\tau \tr{C'}$. If $\tr C
  \trans\tau P$, then there exists $C'$ such that $P = \tr{C'}$. 
\end{thm}

\begin{figure}
\begin{align*}
  \tr {\kam t \stk} & \eqdef \tr t \parallel \out c {\tr \stk} \\
  \tr {\kamC \stk v} & \eqdef \tr \stk \parallel \out \nval {\trv v} \\
  \tr \mts & \eqdef \out b \nil \\
  \tr {\starg t \cons \stk} & \eqdef \Arg{\tr t}{\tr \stk}\\
  \tr {\stfun v \cons \stk} & \eqdef \Fun{\trv v}{\tr \stk}\\
  \Arg {P_t}{P_\stk} & \eqdef \inp \nval x {(P_t \parallel \out c
                       {\Fun x {P_\stk}})} \\
  \Fun {P_v}{P_\stk} & \eqdef P_v \parallel \out c {P_\stk}\\
  \\
  \trv {x} & \eqdef x \\
  \trv {\lam x t} & \eqdef \inp \nval x {\tr t} \\
  \tr {v} & \eqdef \inp c p {(p \parallel \out \nval {\trv v})} \\
  \tr {t \app s} & \eqdef \inp c p {(\tr t \parallel \out c {\Arg {\tr
                   s} p})} 
\end{align*}
\caption{Encoding of the CK machine in \hocore}
\label{fig:ck-hocore}
\end{figure}

\subsection{Bisimilarities} To internalize normal-form bisimilarity,
we distinguish free variables~$\fvr$ from bound variables $x$, as in
Section~\ref{sec:full-abs}. We assume $\lambda$-terms to be
well-formed, so that every bound variable has its binder.
\[
  t, s \bnf  \fvr \midd x \midd \lam x t \midd t \app s \qquad
  v \bnf \fvr \midd \lam x t
\]
The evaluation of an open term may result either in a value in an
empty stack, or in a \emph{stuck term} of the form
$\kamC {\stfun \fvr \cons \stk} v$: because the function is a free
variable, the $\beta$-reduction cannot be triggered. Normal-form
bisimilarity should be defined on these two kinds of normal forms.

For values, a first possibility is to proceed as in call-by-name,
i.e., to equate identical free variables, and to relate the bodies of
$\lambda$-abstractions. This variant (discussed in
Remark~\ref{rem:cbv-as-cbn}) is too discriminating in the plain
$\lambda$-calculus: e.g., it distinguishes a free variable $\fvr$ from
its $\eta$-expansion $\lam x {\fvr \app x}$ although these terms are
contextually equivalent in call-by-value. Lassen~\cite{Lassen:LICS05}
proposes instead to compare values by applying them to a fresh
variable. For instance, applying $\fvr$ and $\lam x {\fvr \app x}$ to
a fresh variable $\fvr'$ produces the same term $\fvr \app \fvr'$ on
both sides after evaluation.

Normal-form bisimilarity compares two stuck terms
$\kamC {\stfun \fvr \cons \stk} v$ and
$\kamC {\stfun \fvr \cons {\stk'}}{v'}$ by relating $v$ and $v'$, and
by testing $\stk$ and $\stk'$ with a fresh variable $\fvr'$. To reflect
this test on stacks, we define normal-form bisimilarity not on terms,
but on configurations.

\begin{defi}
  \label{def:nf-cbv}
  A symmetric relation $\relR$ on configurations is a normal-form
  bisimulation if $C \relR C'$ implies:
  \begin{itemize}
  \item if $C \rtc{\kred} \kamC \mts v$, then there exists $v'$ such
    that $C' \rtc{\kred} \kamC \mts {v'}$ and for all fresh $\fvr$, we
    have 
    $\kamC {\stfun v \cons \mts} \fvr \relR \kamC {\stfun {v'} \cons
      \mts} \fvr$;
  \item if $C \rtc{\kred} \kamC {\stfun \fvr \cons \stk} v$, then
    there exist $\stk'$ and $v'$ such that
    $C' \rtc{\kred} \kamC {\stfun \fvr \cons \stk'}{v'}$, and for all
    fresh $\fvr'$, we have
    $\kamC {\stfun v \cons \mts}{\fvr'} \relR \kamC {\stfun {v'} \cons
      \mts}{\fvr'}$ and
    $\kamC \stk {\fvr'} \relR \kamC {\stk'} {\fvr'}$.
  \end{itemize}
  Normal-form bisimilarity $\nfbisim$, is the largest normal-form bisimulation.
\end{defi}
In both clauses, we compare values by considering configurations of
the form $\kamC {\stfun v \cons \mts} \fvr$, which corresponds to the
application $v \app \fvr$, as wished. Normal-form bisimilarity is
extended to terms so that $t \nfbisim s$ if
$\kam t \mts \nfbisim \kam s \mts$. This equivalence can be related to
B\"ohm tree equivalence up to infinite $\eta$ expansion through a
continuation passing style transformation~\cite{Lassen:LICS05}. It is
not complete \wrt the contextual equivalence of the call-by-value
$\lambda$-calculus~\cite[Example 3.2]{Lassen:LICS05}.

We internalize applicative bisimilarity as in call-by-name
(Section~\ref{sec:full-abstr-applicative}), by considering a
$\lambda$-calculus with de Bruijn indices and an environment-based CK
machine, i.e., the CEK abstract
machine~\cite{Felleisen-Friedman:FDPC3,Felleisen-al:09}. The
definition of the bisimilarity differs from call-by-name only in the
fact that the testing argument must be a
$\lambda$-abstraction. Consequently, we just use the
argument-generating rules from Figure~\ref{fig:am-ab} on top of the
CEK machine, and change the rule $\hname{Restart}$ so that it can
trigger only if the constructed term is a
$\lambda$-abstraction. Internalizing contextual equivalence in
call-by-value simply consists in adding the unmodified rules for
generating a context (Section~\ref{ssec:ctxequiv}) on top of the CEK
machine. In the rest of this section, we discuss only the more
interesting case of the NFB machine and its translation in \hocore.

\subsection{NFB Machine}
We extend the CK machine into a NFB machine in
Figure~\ref{fig:am-cbv}. Configurations include a counter $n$ to
generate fresh variables, as in call-by-name. The first four machine
steps correspond to the CK machine, while the last three ones compare
normal forms.

\begin{figure}
\begin{align*}
  \fnevalcbv{t \app s} \stk n
  & \fnred{\tau} \fnevalcbv t {\starg s \cons \stk} n \tag{\hname{Fun}}\\
  \fnevalcbv v \stk n
  & \fnred{\tau} \fncontcbv \stk v n \tag{\hname{Switch}}\\
  \fncontcbv {\starg t \cons \stk} v n
  & \fnred{\tau} \fnevalcbv t {\stfun v \cons \stk} n \tag{\hname{Arg}}\\
  \fncontcbv {\stfun{\lam x t} \cons \stk} v n
  & \fnred{\tau} \fnevalcbv {\subst v x t} \stk n \tag{\hname{Beta}} 
  \\[3mm]
  \fncontcbv \mts v n & \fnred{\flambda} \fncontcbv {\stfun v \cons \mts}
  n {n+1} \tag{\hname{Val}}\\
  \fncontcbv {\stfun \fvr \cons \stk} v n
  & \fnred\fvr\fnred\fval \fncontcbv {\stfun v \cons \mts}
  n {n+1} \tag{\hname{Stuck-Val}}\\
  \fncontcbv {\stfun \fvr \cons \stk} v n
  & \fnred\fvr\fnred\fctx \fncontcbv \stk n {n+1} \tag{\hname{Stuck-Context}}
\end{align*}
\caption{NFB machine for call by value}
\label{fig:am-cbv}
\end{figure}

If we get a value, we flag $\flambda$ and apply it to a fresh variable
(transition $\hname{Val}$). In the case of a stuck term
$\kamC {\stfun \fvr \cons \stk} v$, we have to choose whether we test
$v$ (transition $\hname{Stuck-Val}$) or $\stk$ (transition
$\hname{Stuck-Context}$). In both cases, we flag the free
variable~$\fvr$. The test for values is the same as in the transition
$\hname{Val}$. If we choose the stack $\stk$, we simply restart the
machine in continuation mode with $\stk$ and a fresh variable.

We prove that the resulting machine equivalence coincides with
normal-form bisimilarity using the same proof technique as in
call-by-name.
\begin{thm}
  $t \nfbisim s$ iff there exists $n > \max{(\fv t \cup \fv s)}$ such
  that $\fnevalcbv t \mts n \ambisim \fnevalcbv s \mts n$.
\end{thm}

\begin{rem}
  \label{rem:cbv-as-cbn}
  In extensions of the $\lambda$-calculus where functions are not the
  only values (e.g., with booleans or integers as primitive
  constructs), a value is no longer equivalent to its $\eta$-expansion
  in general, and we would have to consider the finer normal-form
  bisimilarity which distinguishes values based on their kind. To
  internalize it in the plain $\lambda$-calculus, we would replace
  $\hname{Val}$ with the transitions $\hname{Done}$ and
  $\hname{Lambda}$ below, and $\hname{Stuck-Val}$ with the new step
  below.
  \begin{align*}
    \fncontcbv \mts \fvr n & \fnred{\fvr}\fnred{\fdone} \tag{\hname{Done}}\\
    \fncontcbv \mts {\lam x t} n & \fnred{\flambda} \fnevalcbv {\subst n x t}
                                   \mts {n+1} \tag{\hname{Lambda}}\\
    \fncontcbv {\stfun \fvr \cons \stk} v n
                           & \fnred{\fvr}\fnred{\fval} \fncontcbv \mts v n
                             \tag{\hname{Stuck-Val}}
  \end{align*}
  If we get a free variable, we signal it and we are done. We
  instantiate $\lambda$-abstractions with a fresh variable. Finally,
  the transition $\hname{Stuck-Val}$ goes to $\fncontcbv \mts v n$,
  after which we immediately apply either $\hname{Done}$ or
  $\hname{Lambda}$ depending on $v$.
\end{rem}

We present the encoding of the call-by-value NFB machine in
Figure~\ref{fig:cbv-nfb-hocore}. We follow the same principles as in
Section~\ref{ss:nfb-hocore}. For example, we still use an output on
$b$ to describe what to do when the stack is empty---run the
$\Restart$ process; the message on $b$ is simply discarded when the
stack is not empty. The $\Restart$ process flags $\flambda$ and then
recreates the configuration where the value in $\nval$ is applied to a
fresh variable, represented by the counter~$n$ stored in~$k$.

\begin{figure}
\begin{align*}
  \tr {t \app s} & \eqdef \inp c p {(\tr t \parallel \out c {\Arg {\tr
                   s} p})} \\
  \tr {v} & \eqdef \inp c p {(p \parallel
            \out \nval {\trv v} \parallel \out b \Restart)} \\
  \trv {\lam x t} & \eqdef \inp \nval x {\inp b \osef {\tr t}} \\
  \trv {\fvr} & \eqdef \trint \fvr \\
  \\
  \trint 0  & \eqdef \inp z {\_}{\inp \nval x {\inp c {p} {\inp k y
                    {\inp b \osef {\IChoice{x,p,y}}}}}} \\
  \trint {n+1} & \eqdef \inp {suc}{\_}{\trint{n}} \\
  \Restart & \eqdef  \inp{\flambda}{\_}{\inp \nval x {\inp k y {(
             \Fun x {\tr\mts} \parallel \out \nval y \parallel \out k
             {\inp {suc}{\_} y}
             \parallel \out b \nil
             )}}}\\
  \IChoice{P_v,P_\stk,P_n} & \eqdef \inp {\fenter}{\_}{(\Fun{P_v}{\tr
                             \mts} \parallel \out \nval {P_n}
                             \parallel \out k {\inp {suc} \osef
                             {P_n}} \parallel \out b \Restart )} \\
                 & + \inp{\fskip}{\_}{(P_\stk \parallel \out
                             \nval {P_n} \parallel \out k {\inp {suc}
                             \osef {P_n}} \parallel \out b \Restart )} \\
  \\
  \tr \mts & \eqdef \inp b x x \\
  \tr {\starg t \cons \stk} & \eqdef \Arg{\tr t}{\tr \stk}\\
  \tr {\stfun v \cons \stk} & \eqdef \Fun{\trv v}{\tr \stk}\\
  \Arg {P_t}{P_\stk} & \eqdef \inp \nval x {\inp b \osef {(P_t \parallel \out c
                       {\Fun x {P_\stk}})}} \\
  \Fun {P_v}{P_\stk} & \eqdef P_v \parallel \out c {P_\stk}\\
  \\
  \tr{\fnevalcbv t \stk n} &\eqdef \tr{t} \parallel \out{c}{\tr \stk} \parallel \out k
  {\trint n} \\
  \tr{\fncontcbv \stk v n} &\eqdef {\tr \stk} \parallel \out \nval
                             {\trv v} \parallel \out k {\trint n} \parallel \out b \Restart
\end{align*}
\caption{Encoding of the call-by-value NFB machine into \hocore}
\label{fig:cbv-nfb-hocore}
\end{figure}


A free variable $\fvr$ is observable in a stuck configuration
$\fncontcbv {\stfun f \cons \stk} v n$; in that case, it signals
itself with a sequence of inputs on $suc$ ended by an input on $z$,
and then it non-deterministically chooses between testing $v$ or
$\stk$. The $\IChoice{P_v,P_\stk,P_n}$ process is making that
choice, applying $P_n$ to $P_v$ in the first case, or running $P_n$ in
$P_\stk$ in the second one. Unlike in call-by-name, we do not have to
go recursively through the stack: as a result, although the CK machine
is more complex than the KAM, the translation of the call-by-value NFB
machine is arguably simpler than the call-by-name one. 

We can prove full abstraction between the call-by-value NFB machines
and their translated \hocore processes as in
Section~\ref{ssec:fullabs-nf}, from which we can deduce full
abstraction between \hocore and the $\lambda$-calculus with
normal-form bisimilarity.
\begin{thm}
  $t \nfbisim s$ iff there exists $n > \max{(\fv t \cup \fv s)}$ such that
  $\tr{\fneval t \mts n} \hobisim \tr{\fneval s \mts n}$.
\end{thm}

\section{Control Operators}
\label{sec:control}

The benefit of using abstract machines as an intermediary step is that
our encoding can easily be extended to control operators, the
semantics of which can be expressed using abstract machines. Among the
existing control operators, we choose to encode the
$\lambda\mu$-calculus~\cite{Parigot:LPAR92}, to show how to deal with a calculus
with multiple binders.

\subsection{Extended KAM for the Call-by-Name
  \texorpdfstring{$\lambda\mu$-calculus}{lambda-mu-calculus}}

The $\lambda\mu$-calculus extends the $\lambda$-calculus with
\emph{names} or \emph{continuation variables}---ranged over by $\name$
and $\nameb$---which represent an unknown continuation, and a $\mu$
operator to capture continuations. The syntax of terms becomes as
follows.
\[
  t,s \bnf x \midd \lam x t \midd t \app s \midd \mmu \name \nameb t
\]
The semantics of the $\lambda\mu$-calculus is not defined on plain
terms but on \emph{named terms} of the form $\nt \name t$. In
$\mmu \name \nameb t$, the occurrences of $\name$ in $\nt \nameb t$
are bound. We write $\fn t$ and $\fn {\nt \name t}$ for the set of
free names of a term or named term. A term is closed if it has no free
variable and no free name. A named term $\nt \name t$ cannot be closed
as it has at least $\name$ among its free names.

We define the semantics of the $\lambda\mu$-calculus using an
extension of the KAM, meaning that the continuations captured by the
$\mu$ operator are represented by stacks. Like in the call-by-name
$\lambda$-calculus, a stack represents a sequence of applications of
terms; the difference is that it is ended by a name which represents
the toplevel in which the term is executed.
\[
  \stk \bnf t \cons \stk \midd \name
\]

\begin{figure}
\begin{align*}
  \substc \stk \name \name
  & \eqdef \stk
  & \plug \name t
  & \eqdef \nt \name t
  \\
  \substc \stk \name  {\nameb}
  & \eqdef \nameb \mbox{ if } \name \neq \nameb
  & \plug {(s \cons \stk)} t
  & \eqdef \plug \stk {t \app s}
\end{align*}
\begin{align*}
  \substc \stk \name x
  & \eqdef x \\
  \substc \stk \name {\lam x t}
  & \eqdef \lam x {\substc \stk \name t}
  \\
  \substc \stk \name {(t \app s)}
  & \eqdef \substc \stk \name t \app \substc \stk \name s \\
  \substc \stk \name {\mmu {\nameb_1}{\nameb_2} t}
  & \eqdef \mu {\nameb_1}.\plug {(\substc \stk \name {\nameb_2})}
    {\substc \stk \name t}
\end{align*}
\caption{Structural Substitution}
\label{fig:struct-subs}
\end{figure}

The term $\mmu \name \nameb t$ captures such a stack $\stk$ to replace
$\name$ in $\nt \nameb t$ with it. \emph{Structural
  substitution}~\cite{Ariola-Herbelin:JFP08} $\substc \stk \name t$
produces a plain term where the name $\name$ is replaced by $\stk$ in
$t$. It is defined in Figure~\ref{fig:struct-subs}, alongside two
auxiliary operations: name substitution $\substc \stk \name \name$
which produces a stack, and plugging $\plug \stk t$ which produces a
named term. Plugging simply reconstructs around $t$ the applications
represented by the stack $\stk$. In the case
$\substc \stk \name {\mmu {\nameb_1}{\nameb_2} t}$, we assume the free
name $\name$ to be distinct from $\nameb_1$, which is always possible
using $\alpha$-conversion. The captured stack $\stk$ is restored when
$\nameb_2 = \name$, as we plug $\substc \stk \name t$ inside $\stk$.

The semantics is given by the extended KAM, defined on the same
configurations as in Section~\ref{sec:KAM}.
\begin{align*}
  C & \bnf \kam t \stk \tag{configurations} \\
  \\
  \kam {t \app s} \stk
    & \kred \kam t {s \cons \stk}  \tag{\hname{Push}}\\
  \kam {\lam x t} {s \cons \stk}
    & \kred \kam {\subst s x t} \stk \tag{\hname{Grab}}\\
  \kam {\mmu \name \nameb t} \stk
    & \kred \kam {\substc \stk \name t}{\substc \stk \name \nameb}
      \tag{\hname{Capture}}
\end{align*}
The $\hname{Push}$ and $\hname{Grab}$ steps are unchanged. The last
step captures the stack~$\stk$ and replaces~$\name$ with it, thanks to
the operations defined in Figure~\ref{fig:struct-subs}.  To evaluate a
closed term, we assume the existence of a distinguished name $\tp$
representing the toplevel~\cite{Ariola-Herbelin:JFP08}, and we execute
$\kam t \tp$.

\begin{exa}
  Unlike the $\mu$ operator, the control operator call/cc from Scheme
  leaves the captured stack in place when it triggers. This behaviour
  can be encoded in the $\lambda\mu$-calculus as follows~\cite{Parigot:LPAR92}:
  \[
    \callcc \eqdef \lam x {\mmu \name
      \name {x \app \lam y {\mmu \nameb \name y}}}
  \]
  The captured stack is immediately restored, but also saved in the
  value passed to $x$. To illustrate the extended KAM, we run it on an
  example. Let $id \eqdef \lam x x$, $\Omega \eqdef (\lam x {x \app
    x}) \app (\lam x {x \app x})$, $v \eqdef \lam z {z \app id \app
    \Omega}$, and $t$ be any term. We compute $\callcc \app v \app t$
  in the name $\tp$.
  \begin{align*}
    \kam {\callcc \app v \app t} \tp
    & \kred^2 \kam \callcc {v \cons t \cons \tp} \tag{\hname{Push} -
      \hname{Push}} \\
    & \kred \kam {\mmu \name \name {v \app \lam y {\mmu \nameb \name
      y}}} {t \cons \tp} \tag{\hname{grab}} \\
    & \kred \kam {v \app \lam y {\mmu \nameb \tp {y \app t}}}{t \cons
      \tp} \tag{\hname{Capture}} \\
    & \kred^2 \kam{(\lam y {\mmu \nameb \tp {y \app t}}) \app id \app
      \Omega}{t \cons \tp} \tag{\hname{Push} - \hname{Grab}} \\
    & \kred^3 \kam{\mmu \nameb \tp {id \app t}}{\Omega \cons t \cons
      \tp} \tag{\hname{Push} - \hname{Push} - \hname{Grab}} \\
    & \kred \kam {id \app t} \tp \tag{\hname{Capture}} \\
    & \kred^2 \kam t \tp \tag{\hname{Push} - \hname{Grab}}
  \end{align*}
  After the first capture, the stack $t \cons \tp$ is still in place,
  while the stack $\Omega \cons t \cons \tp$ disappears during the
  second capture, because $\nameb$ does not occur in the body of the
  $\mu$ operator. 
\end{exa}

Using a substitution-based semantics for the $\mu$ operator means it
can easily be encoded in \hocore as a higher-order communication. We
translate names as \hocore variables, chosen to be distinct from the
set of translated term variables. When we evaluate closed terms, the
distinguished name $\tp$ is the end of the stack and is translated as
such. The $\mu$ operator is translated as follows.
\begin{align*}
  \tr \tp & \eqdef \out b \nil \\
  \tr {\mmu \name \nameb t} & \eqdef \inp c \name {(\tr t \parallel \out
    c \nameb)}
\end{align*}
The rest of the translation is the same as in
Section~\ref{sec:KAM}. The translated $\mu$ operator reflects the
$\hname{Capture}$ step by receiving the current stack on $c$. Because
we execute closed terms, either $\nameb=\name$, or $\nameb$ is bound
higher in the term and has been substituted before we get to this
point. The term $\tr t$ is therefore executed in a previously captured
stack.

The correspondence between the extended KAM and its \hocore
translation can be stated the same way as for the regular KAM
(Theorem~\ref{th:cor-KAM-trans}). The $\hname{Capture}$ step
corresponds to exactly one communication in the translation of the machine. 

\subsection{NFB Machine}
\label{ss:nfb-mu}

To internalize normal-form bisimilarity into
the extended KAM, we need to distinguish not only free variables $f$
from bound variables $x$, but also free names~$\fnm$,~$\fnmb$ from bound names
$\bnm$. Indeed, normal-form bisimilarity uses fresh names to compare
terms and discriminates terms based on these names. We change the
syntax of terms as follows.
\[
  \name \bnf \fnm \midd \bnm \qquad 
  t,s \bnf \fvr \midd x \midd \lam x t \midd t \app s \midd \mmu \bnm \name t
\]
We assume terms to be well-formed, i.e., to not contain a bound
variable or bound name without its binder. The name at the end of a
stack is necessarily free.
\[
  \stk \bnf t \cons \stk \midd \fnm
\]

Given two stacks $\stk_1$, $\stk_2$, we write $\stk_1 \relR \stk_2$ if
their elements are pairwise related and they are ended by the same
name. Formally, we have either $\stk_1 = \stk_2 = \fnm$ for some
$\fnm$, or $\stk_1 = t \cons \stk'_1$, $\stk_2 = s \cons \stk'_2$,
$t \relR s$, and $\stk'_1 \relR \stk'_2$. The definition of
normal-form bisimilarity for the call-by-name $\lambda\mu$-calculus is
as follows.
\begin{defi}
  A symmetric relation $\relR$ is a normal-form bisimulation if $t
  \relR s$ and $\fnm$ fresh implies:
  \begin{itemize}
  \item if $\kam t \fnm \rtc\kred \kam {\lam x t'}{\fnmb}$, then there
    exists $s'$ such that $\kam s \fnm \rtc\kred \kam {\lam x
      s'}{\fnmb}$ and $\subst \fvr x {t'} \relR \subst \fvr x {s'}$
    for a fresh $\fvr$; 
  \item if $\kam t \fnm \rtc\kred \kam {\fvr}{\stk}$, then there
    exists $\stk'$ such that $\kam s \fnm \rtc\kred \kam
    {\fvr}{\stk'}$ and $\stk \relR \stk'$.
  \end{itemize}
  Normal-form bisimilarity is the largest normal-form bisimulation. 
\end{defi}

Normal-form bisimilarity is not complete \wrt contextual equivalence
of the $\lambda\mu$-calculus when defined either with weak-head
reduction~\cite{Lassen:99} or head reduction~\cite{Lassen:LICS06}. For
an extended $\lambda\mu$-calculus--- the so-called $\Lambda\mu$
calculus~\cite{Saurin:LICS05}---with head reduction, it coincides with
solvable equivalence~\cite{Barendregt:84,Lassen:LICS06}.

The bisimulation definition is almost the same as in
$\lambda$-calculus (Definition~\ref{def:nf-cbn}) except for
names. Terms are compared within some fresh name $\fnm$ for each
bisimulation round, as using a single name like $\tp$ for all rounds
would not be discriminating enough. It would relate for instance $x$
and $\mmu \bnm \tp x$, because $\kam x \tp$ and
$\kam {\mmu \bnm \tp x} \tp$ evaluate to $\kam x \tp$, but they
behave differently: the second term discards its continuation while the
first one does not. Similarly, we expect the resulting normal forms to
exhibit a common name $\fnmb$ (not necessarily equal to $\fnm$). Terms
evaluating to normal forms with distinct toplevel names would imply
that they return to different continuations, and a surrounding context
could easily distinguish them by picking a terminating continuation
for a term and a diverging continuation for the other.

We adapt the $\lambda$-calculus NFB machine to take these changes into
account, resulting in the machine of Figure~\ref{fig:nfbm-mu}. In the
machine, free names are natural numbers, and configurations
$\kmunfev t \stk n m$ include an extra counter $m$ to generate fresh
names. This counter is used as a new stack whenever we restart a
computation, in the $\hname{Lambda}$ and $\hname{Enter}$ steps. The
$\hname{Lambda}$ and $\hname{Done}$ steps now flag the free name
$\fvr$ of the resulting normal forms. The rest of the machine is the
same as in $\lambda$-calculus, except for the $\hname{Capture}$
step.

\begin{figure}
\begin{align*}
  \kmunfev {t \app s} \stk n m
  & \fnred\tau \kmunfev t {s \cons \stk} n m \tag{\hname{Push}}\\
  \kmunfev {\lam x t} {s \cons \stk} n m
  & \fnred\tau \kmunfev {\subst s x t} \stk n m \tag{\hname{Grab}}\\
  \kmunfev {\mmu \bnm \name t} \stk n m
  & \fnred\tau \kmunfev {\substc \stk \bnm t} {\substc \stk \bnm
    \name} n m \tag{\hname{Capture}} \\
  \kmunfev {\lam x t} \fnm n m
  & \fnred{\fnm} \fnred \flambda \kmunfev {\subst n x t} m
    {n+1}{m+1} \tag{\hname{Lambda}}  \\
  \kmunfev \fvr \stk n m
  & \fnred\fvr \kmunfcont \stk n m \tag{\hname{Var}}  \\
  \\
  \kmunfcont \fnm n m & \fnred \fnm \fnred \fdone \tag{\hname{Done}}  \\
  \kmunfcont {t \cons \stk} n m
  & \fnred\fenter \kmunfev t m n {m+1} \tag{\hname{Enter}} \\
  \kmunfcont {t \cons \stk} n m
  & \fnred\fskip \kmunfcont \stk n m \tag{\hname{Skip}} 
\end{align*}
\caption{NFB machine for the call-by-name $\lambda\mu$-calculus}
\label{fig:nfbm-mu}
\end{figure}

The translation of the NFB machine for the $\lambda\mu$-calculus into
\hocore follows the same principles as in $\lambda$-calculus. We just
need to be careful to distinguish between numbers representing free
variables from those standing for free names---by using distinct
\hocore names $suc_{\mathsf v}$, $z_{\mathsf v}$ and
$suc_{\mathsf n}$, $z_{\mathsf n}$. We can prove the equivalence
between normal-form bisimilarity and its machine, and between the NFB
machine and its translation as in Section~\ref{sec:full-abs}.

\begin{thm}
  $t \nfbisim s$ iff there exists $n > \max{(\fv t \cup \fv s)}$ and
  $m > \max{(\fn t \cup \fn s)}$ such that
  $\tr{\kmunfev t m n {m+1}} \hobisim \tr{\kmunfev s m n {m+1}}$.
\end{thm}

\subsection{Environment-based Abstract Machine} As in
$\lambda$-calculus (Section~\ref{sec:full-abstr-applicative}),
internalizing applicative bisimilarity requires to generate a testing
argument. Because we cannot generate binders in the \hocore
translation, we consider a de Bruijn syntax for terms and an
environment-based machine. However, unlike in $\lambda$-calculus,
named terms are not completely closed, because of the toplevel name,
and when we play the bisimulation game, considering a single free name
$\tp$ is not enough, as it would not be sound, like with normal-form
bisimilarity (cf Section~\ref{ss:nfb-mu}). We therefore need to
consider three distinct families of de Bruijn indices: term variables
$n$, bound names $\bnm$, and free names $\fnm$. The de Bruijn syntax
of the $\lambda\mu$-calculus is therefore as follows.
\[
  \name \bnf \bnm \midd \fnm \qquad  t,s \bnf n \midd \lamB t \midd t
  \app s \midd \muB \name t  
\]

We present in Figure~\ref{fig:env-ck} the environment-based extended
KAM. It is similar to the Streicher and Reus
machine~\cite{Streicher-Reus:JFP98}, except their syntax uses names
and not de Bruijn indices, and they use a single environment for term
variables and bound names. Instead, we distinguish the environment
mapping term variables to closures (ranged by $\envt$, $\envtt$) from
the one mapping bound names to stacks (ranged by $\envb$,
$\envbb$). An extra environment mapping free names to stacks is added
when we internalize applicative bisimilarity
(Section~\ref{ss:ab-mu}). Closures~$\clos$ are triples
$(t, \envt, \envb)$, and stacks $\stk$ are composed of closures and
ended by a free or bound name.

\begin{figure}
\begin{align*}
  \envt & \bnf \clos \cons \envt \midd \mtet \tag{term-variables environments}\\
  \envb & \bnf \stk \cons \envb \midd \mteb \tag{bound-names environments}\\
  \clos & \bnf (t, \envt, \envb) \tag{closures}\\
  \stk & \bnf \clos \cons \stk \midd \name \tag{stacks}\\
  C & \bnf \kmuB t \envt \envb \stk \midd \kmuBb \bnm \envb
      \clos \tag{configurations}
  \\
  \\
  \kmuB {t \app s} \envt \envb \stk 
        & \kBred\tau \kmuB t {\envt} \envb
          {(s, \envt, \envb) \cons \stk} \tag{\hname{Push}} \\
  \kmuB 0 {(t, \envt, \envb) \cons \envtt}{\envbb} \stk 
        & \kBred\tau \kmuB t \envt \envb \stk \tag{\hname{Zero-Term}}\\
  \kmuB {n+1} {\clos \cons \envt}{\envb} \stk 
        & \kBred\tau \kmuB n {\envt} \envb \stk
          \tag{\hname{Env-Term}} \\
  \kmuB {\lamB t} \envt \envb {\clos \cons \stk}
        & \kBred\tau \kmuB t {\clos \cons \envt} \envb \stk
          \tag{\hname{Grab}} \\
  \kmuB {\muB \name t} \envt \envb \stk 
        & \kBred\tau \kmuB t \envt {\stk \cons \envb} \name
          \tag{\hname{Capture}} \\
  \kmuB {\lamB t} \envt \envb \bnm
        & \kBred\tau \kmuBb \bnm \envb {(\lamB t, \envt, \envb)}
          \tag{\hname{Restore}}\\
  \\
  \kmuBb 0 {\stk \cons \envbb}{(t, \envt, \envb)}
        & \kBred\tau \kmuB t \envt \envb \stk \tag{\hname{Zero-BName}} \\
  \kmuBb {\bnm+1} {\stk \cons \envb} \clos
        & \kBred\tau \kmuBb \bnm \envb \clos \tag{\hname{Env-BName}} \\
\end{align*}
\caption{Environment-based KAM for the $\lambda\mu$-calculus}
\label{fig:env-ck}
\end{figure}

The first four steps are as in the $\lambda$-calculus. The
$\hname{Capture}$ step replaces the current stack~$\stk$ with
$\name$ and puts $\stk$ on top of the bound-name environment
$\envb$. Whenever a $\lambda$-abstraction is executed within a stack
composed only of a bound name $\bnm$, we restore the previously
captured stack, by looking into $\envb$. It is the role of the
$\hname{Zero-BName}$ and $\hname{Env-BName}$ steps, which operate on
dedicated configurations of the form $\kmuBb \bnm \envb \clos$. 

\subsection{AB machine}
\label{ss:ab-mu}

The definition of the sound and complete applicative bisimilarity for
the $\lambda\mu$-calculus~\cite{Biernacki-Lenglet:MFPS14} is more
intricate than in $\lambda$-calculus, because we compare named
terms. We first recall its definition~\cite{Biernacki-Lenglet:MFPS14}
and using named binders, before moving to de Bruijn indices, and then
to an AB machine.

\subsubsection{Informal definition with named binders} Assume two
named terms $\nt \name t$ and $\nt \name s$ where $t$ and $s$ are
closed. If they evaluate to respectively $\nt \name {\lam x {t'}}$ and
$\nt \name {\lam x {s'}}$, a surrounding context cannot test them by
applying them to an arbitrary closed term $t''$, as
$(\nt \name {\lam x {t'}}) \app t''$ is not a valid
$\lambda\mu$-term. The simplest named term that can be built to lead
to such an application is
$\nt \nameb {(\mmu \name \name {\lam x {t'}}) \app t''}$ for a fresh
$\nameb$. Indeed, the $\mu$-binder captures the named context
$\nt \nameb {\mtctx \app t''}$ (represented as a stack
$t'' \cons \nameb$) to reduce to
$\nt \nameb (\lam x {\substc {t'' \cons \nameb} \name {t'}) \app
  t''}$; the structural substitution $\substc {t'' \cons \nameb} \name
{t'}$
is necessary because $\name$ may occur in $t'$. The $\beta$-reduction
is now possible and the previous named term reduces to
$\nt \nameb {\subst {t''} x {\substc {t'' \cons \nameb} \name
    {t'}}}$. We get the same sequence of reductions with $\nt \name
s$, so that applicative bisimilarity compares
\[
  \nt \nameb {\subst {t''} x {\substc {t'' \cons \nameb} \name {t'}}}
  \mbox{ and }
  \nt \nameb {\subst {t''} x {\substc {t'' \cons \nameb} \name {s'}}}.
\]

We need to write such a test with de Bruijn indices and
environments. We represent the successive fresh names $\name$,
$\nameb$, $\ldots$ used by the bisimilarity with indices 0, 1,
$\ldots$ Assume we evaluate the closed closures $(t, \envt, \envb)$
and $(s, \envtt, \envbb)$ in the free name $0$, resulting respectively
into $(\lamB {t'}, \envt', \envb')$ and
$(\lamB {s'}, \envtt', \envbb')$ in the same name $0$. Let $t''$ be a
closed term. The term substitution $\subst {t''} x \cdot$ is
represented as in plain $\lambda$-calculus, by extending $\envt'$ and
$\envtt'$ with the closure corresponding to $t''$. To represent the
structural substitution $\substc {t'' \cons \nameb} \name \cdot$ we
introduce an extra environment dedicated to free names. 
 
\subsubsection{Testing environments} 
A \emph{free-names environment} $\envf$ is a mapping from free names
(represented as de Bruijn indices) to stacks, defined as follows.
\[
  \envf \bnf \stk \cons \envf \midd \mtef \tag{free-names environments}
\]
This environment is extended by the bisimilarity each time it tests
pairs of named values. We write $\empclos {t''}$ for the closure
$(t'', \mtet, \mteb, \mtef)$ used in these tests. The stack
$t'' \cons \nameb$ is represented as $\empclos{t''} \cons 1$, and
$\envf$ maps $0$ to $\empclos{t''} \cons 1$. A subsequent test would
extend $\envf$ to map $1$ to $\empclos{t'''} \cons 2$ for some $t'''$,
then a following test would map $2$ to another closure, \etc

We see that the bisimilarity extends $\envf$ at the end, and we write
$\envf \snoc \stk$ for such an extension. We also notice that the
resulting $\envf$ maps each free name to a stack containing only one
element. More precisely, an environment $\envf$ generated by
successive bisimilarity tests is either empty, or its $i^{\text{th}}$
element is of the form $\empclos t \cons i$ for some closed
term~$t$. We say such an environment is a \emph{testing environment},
of rank $\fnm$ if it has $\fnm$ elements. A testing environment of
rank $\fnm>0$ has only $\fnm$ as free-name index without a mapping, at
the end of its last stack. For example, a testing environment of rank
2 is of the form $\stk_0 \cons \stk_1 \cons \mtef$, mapping $0$ to
$\stk_0 = \empclos{t_0} \cons 1$ for some closed $t_0$, and $1$ to
$\stk_1 = \empclos{t_1} \cons 2$ for some closed $t_1$.

\subsubsection{Applicative bisimilarity}
We extend closures to include testing environments
$\clos \bnf (t, \envt, \envb, \envf)$, and we say that $\clos$ is of
rank $\fnm$ when its testing environment is of rank $\fnm$. A stack is
closed if it is composed only of closed closures. A closure is closed
if $\envt$ contains more elements than the highest term-variable index
in $t$, $\envb$ contains more elements than the highest bound-name
index in $t$, $\envt$ contains only closed closures, and $\envb$ only
closed stacks.

We add the testing environment $\envf$ to configurations, and we
define steps to look in~$\envf$ for the stack corresponding to a given
free-name index. We discuss these changes in details when defining the AB
machine. We adapt the definition of applicative
bisimilarity~\cite{Biernacki-Lenglet:MFPS14} to our setting as
follows.

\begin{defi}
  A symmetric relation $\relR$ on closed closures is an applicative
  bisimulation if
  $(t, \envt, \envb, \envf) \relR (s, \envtt, \envbb, \envf)$ with
  $\envf$ of rank $\fnm$ and
  $\kmuBab t \envt \envb \envf \fnm \kBredrtc\tau \kmuBab {\lamB {t'}}
  {\envt'}{\envb'} \envf \fnm$ implies that there exists $s'$,
  $\envtt'$, and $\envbb'$ such that
  $\kmuBab s {\envtt}{\envbb} \envf \fnm \kBredrtc\tau \kmuBab {\lamB
    {s'}}{\envtt'}{\envbb'} \envf \fnm$ and for all closed $t''$, we
  have
  \[
  (t', \empclos{t''} \cons \envt', \envb', \envf \snoc
  (\empclos{t''} \cons \fnm+1)) \relR \\ (s', \empclos{t''} \cons
  \envtt', \envbb', \envf \snoc (\empclos{t''} \cons \fnm+1)).
  \]
  Applicative bisimilarity $\appbisim$ is the largest applicative bisimulation. 
\end{defi}

Applicative bisimilarity executes closed closures with the same
testing environment of rank $\fnm$ within the stack composed only of
the free name $\fnm$. It then compares the resulting values by
considering a closed closure $\empclos{t''}$ to instantiate the
$\lambda$-abstractions, and to extend the testing environment which
becomes of rank $\fnm+1$.

\subsubsection{AB machine} We present the most interesting steps of
the corresponding AB machine in Figure~\ref{fig:abm-mu}. Evaluation
configurations $\kmuBab {t} \envt \envb \envf \stk$ now include the
testing environment~$\envf$, which is passed along unchanged in the
steps corresponding to the extended KAM (omitted in
Figure~\ref{fig:abm-mu}).

Whenever a $\lambda$-abstraction is executed in a stack composed only
of a free name $\fnm$, we have two possibilities: either~$\fnm$ is
mapped in $\envf$, otherwise we have to generate the closure~$\empclos
t$ of the bisimulation test. The configuration $\kmuBabf \fnm \envf
{\envf_r} \clos {\fnm'}$ is designed with these two cases in mind. Its
main role is to look for $\fnm$ in $\envf$ and to restore the
corresponding stack if the mapping exists.

\begin{figure}
\begin{align*}
  \kmuBab {t \app s} \envt \envb \envf \stk
        & \kBred\tau \kmuBab t {\envt} \envb
          \envf {(s, \envt, \envb, \envf) \cons \stk} \tag{\hname{Push}}\\
  & \hspace{0.9em}\vdots \tag{cf. Figure~\ref{fig:env-ck}}\\
  \kmuBab {\lamB t} \envt \envb \envf \fnm
        & \kBred\tau \kmuBabf \fnm \envf \mtef {(\lamB t, \envt, \envb,
          \envf)} \fnm \tag{\hname{Lookup-FName}}\\
  \\
  \kmuBabf 0 {\stk \cons \envf'} {\envf_r} {(t, \envt, \envb,
  \envf)} \fnm
        & \kBred\tau \kmuBab t \envt \envb \envf \stk
          \tag{\hname{Zero-FName}} \\
  \kmuBabf {\fnm+1} {\stk \cons \envf} {\envf_r} \clos {\fnm'}
        & \kBred\tau \kmuBabf \fnm \envf {\stk \cons {\envf_r}} \clos
          {\fnm'} \tag{\hname{Env-FName}} \\
  \kmuBabf 0 \mtef {\envf_r} \clos {\fnm}
        & \kBred\flambda \kmuBabint 0 1 0 \mtr \clos {\envf_r}  {\fnm}
          \tag{\hname{Arg}} \\
  \\
  \kmuBabint n {\acc}{\acc_\mu} \rho \clos {\envf_r} \fnm
        & \kBred{\fsuc} \kmuBabint {n+1}{\acc+1}{\acc_\mu} \rho \clos {\envf_r} \fnm
          \tag{\hname{Suc}}\\
  & \hspace{0.9em} \vdots \tag{cf. Figure~\ref{fig:am-ab}} \\
  \kmuBabtm t {\acc} {\acc_\mu} \rho \clos {\envf_r} \fnm
  & \kBred{\findmu} \kmuBabintmu 0 t {\acc} {\acc_\mu} \rho \clos
    {\envf_r} \fnm \tag{\hname{Ind-Mu}} \\
  \kmuBabintmu \bnm t {\acc} {\acc_\mu} \rho \clos
  {\envf_r} \fnm 
  & \kBred{\fsucmu} \kmuBabintmu {\bnm+1} t {\acc} {\acc_\mu} \rho \clos
    {\envf_r} \fnm \tag{\hname{Suc-Mu}} \\
  \kmuBabintmu \bnm t {\acc} {\acc_\mu} \rho \clos
  {\envf_r} \fnm
  & \kBred{\fmu} \kmuBabtm {\muB \bnm t} {\acc} {\maxx{\bnm+1}{\acc_\mu}
    -1} \rho \clos {\envf_r} \fnm \tag{\hname{Mu}} \\
  \\
  \kmuBabtm t 0 0 \mtr \clos {\envf_r} \fnm
    & \kBred{\frev} \kmuBabefn {\envf_r} {(\empclos t \cons \fnm+1) \cons \mtef} 
      \clos {\empclos t \cons \fnm+1} \tag{\hname{Init-Rev}} \\
  \kmuBabefn {\stk \cons {\envf_r}} \envf \clos \stk
    & \kBred{\frev} \kmuBabefn {\envf_r} {\stk \cons \envf} \clos \stk
      \tag{\hname{Rev}} \\
  \kmuBabefn \mtef \envf {(\lamB s, \envt, \envb, \envf')} \stk
    & \kBred\fdone \kmuBab {\lamB s} \envt \envb
      \envf \stk \tag{\hname{Restart}}
\end{align*}
\caption{AB machine for the call-by-name $\lambda\mu$-calculus}
\label{fig:abm-mu}
\end{figure}

If $\fnm$ is not mapped, then it is stored as the last parameter of
the $\mathsf{fn}$ mode in the $\hname{Lookup-FName}$ step, and passed
along during the argument generation process. Its successor is needed
to create the stack $\empclos t \cons {\fnm +1}$ ($\hname{Init-Rev}$
step).

The second-to-last parameter $\envf_r$ of the $\mathsf{fn}$ mode is
used to compute the reverse of~$\envf$. Indeed, we remind that
applicative bisimilarity extends $\envf$ from its tail. To do so, we
reverse~$\envf$, add the generated closure, and reverse the result
again. We set $\envf_r$ to empty in the $\hname{Lookup-FName}$ step,
and then the elements of $\envf$ are pushed on $\envf_r$ when we go
through $\envf$ in the $\hname{Env-FName}$ step. As a result, when
$\envf$ is empty in the $\hname{Arg}$ transition,~$\envf_r$ contains
the original environment in reverse.

The $\hname{Arg}$ step initiates the argument generation. It triggers
when $\envf$ is empty, and we know that the index at this point is
$0$, because the only free-name index without a mapping allowed by
design is the rank of $\envf$, i.e., its size. As in
Section~\ref{ssec:app-machine}, we start argument generation with the
leftmost de Bruijn index representing a term variable, and then
generate the term from left to right. Compared to $\lambda$-calculus,
the $\mathsf{ind}$ and $\mathsf{tm}$ configurations carry three extra
parameters: the already discussed $\envf_r$ and $\fnm$, and
$\acc_\mu$, which counts how many $\mu$-binders are necessary for the
generated term to be closed. In Figure~\ref{fig:abm-mu}, we omit the
steps that are the same as in Figure~\ref{fig:am-ab} up to these three
parameters.

We generate a term $\muB \bnm t$ by starting with the bound name index
$\bnm$ ($\hname{Ind-Mu}$ and $\hname{Suc-Mu}$ steps). We then add the
$\mu$-binder and update the $\acc_\mu$ counter accordingly
($\hname{Mu}$ step).

After generating the argument $t$, we step to a configuration of the
form $\kmuBabefn {\envf_r} \envf \clos \stk$, the goal of which is to
construct the extended testing environment. The $\hname{Init-Rev}$
step constructs the bottom of the extended environment and remembers
$\empclos t \cons \fnm+1$ as the new stack. The $\hname{Rev}$ step
goes through $\envf_r$ and pushes its elements on $\envf$ to
reconstruct the original environment on top of its new bottom. When
$\envf_r$ is empty, the $\hname{Restart}$ step restarts the machine,
executing $\lamB s$ in the stack $\empclos t \cons \fnm+1$: the next
step is a $\beta$-reduction which produces the configuration we want.\\

\subsubsection{Full abstraction} As in $\lambda$-calculus, the AB
machine flags each step of its argument generation process, and for
each generated argument $t$ there is a corresponding sequence of flags
$\Seq t$. We can prove as in Section~\ref{ssec:app-machine} that
applicative bisimilarity and the AB machine equivalence
coincide. Besides, the AB machine relies on ingredients and operations
(counters, max, \etc) that we know how to translate in \hocore in a
faithful and deterministic way (see Appendix~\ref{app:abmu}), and for
which we can prove the
following full abstraction theorem.

\begin{thm}
  For all closed closures $(t, \envt, \envb, \envf)$ and $(s, \envtt,
  \envbb, \envf)$ of rank $\fnm$, $(t, \envt, \envb, \envf) \appbisim
  (s, \envtt, \envbb, \envf)$ iff $\tr{\kmuBab t \envt \envb \envf
    \fnm} \hobisim \tr{\kmuBab s {\envtt}{\envbb} \envf \fnm}$.
\end{thm}

Full abstraction can also be achieved by internalizing contextual
equivalence directly. The resulting machine would extend the machine
of Figure~\ref{fig:am-ctx} with the steps generating $\mu$-binders of
Figure~\ref{fig:abm-mu}.

\section{Conclusion and Future Work}
\label{sec:conclusion}

We propose encodings of the call-by-name and call-by-value
$\lambda$-calculus with or without control operators into \hocore,
fully abstract w.r.t. normal-form and applicative bisimilarities, and
contextual equivalence. This shows that a minimal higher-order
calculus with a fixed number of hidden names, which is much less
expressive than the name-passing $\pi$-calculus, still has enough
expressive power to faithfully encode these calculi. 

We use abstract machines not only to fix the reduction
strategy, but also as an intermediary step between the
$\lambda$-calculus and \hocore. We turn the equivalences of the
$\lambda$-calculus, and their potentially complex testing conditions,
into a first-order bisimilarity over an LTS (a flag-generating
machine), which is closer to the \hocore equivalence. We believe this
internalization technique can be applied to any language for which an
abstract machine has been defined. No matter how intricate the
bisimilarities for such a language are, it should be always possible
to generate a context as in Section~\ref{ssec:ctxequiv} to internalize
contextual equivalence.

The encodings of the extended abstract machines into \hocore rely on
the same principles, e.g., to represent stacks, non-deterministic
choice, case analyses on terms, \etc We believe it is possible to
automatically derive the encoding from an abstract machine so that
the generated translation verify Lemmas~\ref{lem:determinist}
and~\ref{lem:faithful}, giving us Theorem~\ref{th:am-ho-nf} for free.

Our encodings are \emph{weak}
compositional~\cite{Gorla:IaC10,Parrow:ENTCS06}, i.e., they consist of
a compositional translation of the $\lambda$-term inside a fixed
process representing the machine. We conjecture there is no fully
abstract compositional encoding of the $\lambda$-calculus into
\hocore. In such an encoding, a translated $\lambda$-abstraction could
be composed with an argument (a translated $\lambda$-term) to
$\beta$-reduce. It would require at least one communication on a
public channel name, but the $\lambda$-abstraction should still be
able to protect itself from unwanted behaviours from the outside. It
seems difficult if not impossible to achieve with only hidden names
and no name restriction.

As demonstrated in other
settings~\cite{Prebet:ICALP22,Jaber-Sangiorgi:CSL22}, our encodings
could be useful to import proof techniques such as up-to
techniques~\cite{Sangiorgi-Rutten:12} from the process calculus world into the
$\lambda$-calculus world. It would be interesting especially for
applicative bisimilarity for which powerful up-to techniques have not
been defined yet.

Finally, we would like to explore further the expressiveness of
process calculi by translating the $\lambda$-calculus without a
predefined reduction strategy. The encodings defined so far in the
literature (cf. Section~\ref{sec:intro}) assumes a given strategy for
the $\lambda$-calculus. An exception is the work by Cai and
Fu~\cite{Cai-Fu:MSCS11}, where a $\lambda$-term is represented as a
tree, and $\beta$-reduction is a transformation on trees. It relies on
an unbounded number of restricted names to represent tree nodes; we
wonder if we can use the same ideas with only a fixed number of
names. However, such an encoding works for any source language, as it
is a manipulation of syntax trees. We would like the encoding to be
more tailored to the $\lambda$-calculus, to tell us more about the
relationship between the $\lambda$-calculus and process calculi.

\section*{Acknowledgement}
We would like to thank the anonymous reviewers for their helpful
comments. This work has been partially funded by PHC Polonium and by
the National Science Centre of Poland under grants
no. 2014/15/B/ST6/00619 and 2019/33/B/ST6/00289. Damien Pous was
supported by the LABEX MILYON (ANR-10-LABX-0070) of Université de
Lyon, within the program ``Investissements d'Avenir'' (ANR-11-IDEX-
0007) operated by the French National Research Agency (ANR).

\bibliographystyle{alphaurl}
\bibliography{bibrefs}

\newpage

\appendix

\section{Translation of the \texorpdfstring{$\lambda\mu$}{lambda-mu}-calculus
  AB Machine into \hocore}
\label{app:abmu}

The translation of the AB machine for the call-by-name
$\lambda\mu$-calculus extends the one for the plain
$\lambda$-calculus of Section~\ref{ssec:trans-AM-AB}. We omit the
parts of the translation which are the same or similar to the plain
$\lambda$-calculus one.

The translation of the evaluation mode is as follows. A closure is now
composed of four messages. Because the closures are supposed to be
closed, the process representing
the empty environments~$\mtet$ and~$\mteb$ should never be executed, and we choose
to represent them with $\nil$. The definitions of $\tr \mtef$ and $\tr \name$ are
discussed later.
\begin{align*}
  \tr{\kmuBab t \envt \envb \envf \stk}
  & \eqdef \tr t \parallel \out{\et}{\tr e} \parallel \out{\eb}{\tr \envb}
    \parallel
    \out{\ef}{\tr \envf} \parallel \out c {\tr \stk} \parallel \Prec \\
  \tr {(t, \envt, \envb, \envf)}
  & \eqdef \out{\clos_1}{\tr t} \parallel \out{\clos_2}{\tr \envt}
    \parallel \out{\clos_3}{\tr \envb} \parallel \out{\clos_4}{\tr
    \envf}\\
  \tr {\clos \cons \envt}
  & \eqdef \out \he {\tr \clos} \parallel \out {\et}{\tr \envt}
  & \tr \mtet & \eqdef \nil \\
  \tr {\stk \cons \envb}
  & \eqdef \out \hb {\tr \stk} \parallel \out {\eb}{\tr \envb}
  & \tr \mteb & \eqdef \nil \\
  \tr {\stk \cons \envf}
  & \eqdef \out \hf {\tr \stk} \parallel \out {\ef}{\tr \envf}
  & \tr \mtef & \eqdef \ldots \\
  \tr {\clos \cons \stk}
  & \eqdef \out \hp {\tr \clos}
    \parallel \out c {\tr \stk}
  &  \tr \name & \eqdef \ldots 
\end{align*}

The representations of $\lambda$-abstractions, applications and indices
are the same as in plain
$\lambda$-calculus, adapted to the fact that closures are composed of
four elements now. A $\mu$ operator captures the current stack and
extend the bound-names environment $\envb$ accordingly.
\[
  \tr {\muB \name t}  \eqdef \Mueval {\tr \name}{\tr t}
  \qquad \Mueval {P_\name}{P_t} \eqdef \inp c x {\inp \eb y {(P_t
    \parallel \out \eb {\out \hb x \parallel \out \eb y} \parallel
                             \out c {P_\name})}}
\]

The translation $\tr \bnm$ implements the mode $\mathsf{bn}$ resolving
bound names. The translation $\trb \bnm$ implements the search of
$\bnm$ in $\envb$, while $\tr \bnm$ make a save of $\envb$ before
starting the search. The environment $\envb$ is restored by $\trb 0$
when the search is over.

\begin{align*}
  \tr \bnm & \eqdef \Bound{\trb \bnm} & \Bound{P_\bnm}
  & \eqdef \inp \eb x {(P_\bnm \parallel \out \eb x \parallel \out
    {\svb} x)} \\
  \trb{\bnm + 1} & \eqdef \IndMu{\trb \bnm} & \IndMu{P_\bnm}
  & \eqdef \inp \eb x {(x \parallel \inp \hb \osef {P_\bnm})}\\
           & & \trb 0
  & \eqdef \inp \eb x {(x \parallel \inp \hb y {\inp \eb \osef {\inp
    \svb z {(y \parallel \out \eb z)}}})}
\end{align*}

Similarly, the translation $\tr \fnm$ implements the mode
$\mathsf{fn}$ resolving free names. The translation $\trf \fnm$
implements the search of $\fnm$ in $\envf$, while $\tr \fnm$ make a
save of $\envf$ (on $\svf$) and $\fnm$ (on $\svfn$) before starting
the search. We compute the reverse of $\envf$ on the name $\rev$ using
the processes $\MtRev$ and $\rawConsRev$, which
basically encode a stack. If the search succeeds, the messages on
$\rev$ on $\svfn$ are not useful and are discarded, and the save of
$\envf$ is restored.

\begin{align*}
  \MtRev & \eqdef \inp {mtrev} x {\inp{csrev} \osef x} \\
  \ConsRev{P_\stk}{P_\envf}  & \eqdef \out \hrev {P_\stk} \parallel
                                \out \rev {P_\envf} \parallel \inp
                                {mtrev} \osef {\inp {csrev} x x}
\end{align*}

\begin{align*}
  \tr \fnm & \eqdef \Free{\trf \fnm} & \Free{P_\fnm}
  & \eqdef \inp \ef x {(P_\fnm \parallel \out \ef x \parallel \out
    {\svf} x \parallel \out \svfn {P_\fnm} \parallel \out \rev {\MtRev} )} \\
  \trf{\fnm + 1} & \eqdef \IndFree{\trf \fnm} & \IndFree{P_\fnm}
  & \eqdef \inp \ef x {\inp \rev y {(x \parallel \inp \hf z {(P_\fnm
    \parallel \out \rev {\ConsRev z y})})}}\\
           & & \trf 0
  & \eqdef \inp \ef x {(x \parallel \inp \hf y {\inp \ef \osef {\inp
    \svf z {\inp \svfn \osef {\inp \rev \osef (y \parallel \out \ef z)}}}})}
\end{align*}
If the search fails, it means that $\fnm$ is not mapped to a stack, and $\tr
\mtef$ should start generating the argument. It first discards the
message on $\svf$ which is no longer useful and sets up the
$\mathsf{ind}$ mode.
\[
  \tr \mtef \eqdef \inp \farg \osef {\inp \svf \osef {(
      \out {ind}{\tr 0} \parallel \out k {\tracc 1} \parallel \out
      {\km}{\trkm 0} \parallel \out r {\tr \mtr} \parallel
      \out{initInd} \nil)}}
\]
The message on $\km$ is for the counter $\acc_\mu$, and the
translation $\trkm {\acc_\mu}$ is the same as $\tracc{\acc}$, but
using different names $zerom$ and $sucm$.  The $\RecTm$ process
generating terms needs an extra case analysis on $\acc_\mu$, as the
term generation process may end only if $\acc_\mu = 0$. The process
$\RecTm$ of Section~\ref{ssec:trans-AM-AB} is modified to do this case
analysis when $\acc = 0$.
\begin{multline*}
  \RecTm \eqdef \ldots \\
   \overline{zero}\left\langle \km(u).\left( 
      \begin{aligned}
        u & \parallel \overline {zerom} \left\langle
          \begin{aligned}
            & (\Lambdaa {z} \parallel \out r {\tr \mtr} \parallel \out
            \km u) + (\AppFun z {\tr \mtr} \parallel \out \km u) \\
            &  + (\Muu \parallel \out r {\tr \mtr} \parallel \out k z \parallel \out \km u)+  \DoneTm
          \end{aligned}
        \right\rangle\\
        & \parallel \overline{sucm} \left\langle sukm(\osef).\left(
              \begin{aligned}
            & (\Lambdaa {z} \parallel \out r {\tr \mtr} \parallel \out
            \km u)  \\
            + & (\AppFun z {\tr \mtr} \parallel \out \km u) \\
            + & (\Muu \parallel \out r {\tr \mtr}  \parallel \out k z \parallel
            \out \km u)
              \end{aligned}
              \right)\right\rangle
          \end{aligned}
        \right)\right\rangle 
\end{multline*}
Generating a de Bruijn index for a term variable, a
$\lambda$-abstraction, or an application is done the same way as in
plain $\lambda$-calculus: the processes $\RecInt$, $\Succ$, $\Var$,
$\rawLambdaa$, $\rawAppFun$, and $\rawApp$ are the same as in
Section~\ref{ssec:trans-AM-AB}.

The process $\Muu$ generates a $\mu$ binder, starting with the bound
name $\bnm$, simulating the $\mathsf{indMu}$ mode. It initiates a
counter on $km2$, which is increased in parallel with $\bnm$. That
counter is then compared in $\MuTm$ with the one on $\km$ to compute
the maximum between them, using the same technique as in $\rawApp$ but
on different names. The processes $\RecIMu$ and $\SuccMu$ behave like
their respective counterpart $\RecInt$ and $\Succ$.

\begin{align*}
  \Muu & \eqdef \inp {\findmu} \osef {(\out {indMu}{\trb 0} \parallel \out
         {\km 2}{\trkm{1}} \parallel \out{initIMu} \nil)} \\
  \RecIMu & \eqdef \inp {initIMu} \osef {\inp {recimu} x {(x
            \parallel \out {recimu} x \parallel \SuccMu + \MuTm)}} \\
  \SuccMu & \eqdef \inp \fsucmu \osef {\inp {indMu} x {\inp {\km 2} y
            {(\out {indMu}{\IndMu x} \parallel \out {\km 2}{\Sukm y}
            \parallel \out {initIMu} \nil)}}} \\
  \MuTm & \eqdef \fmu(\osef).tm(x).\km(y_1).\km 2(y_2).indMu(z). \\
       & \quad \left(
         \begin{aligned}
           & \out{maxm1}{y_1} \parallel \out{maxm2}{y_2} \parallel
           \out{initm1}{y_1} \parallel \out{initm2}{y_2} \parallel \\
           & resu(y).\left(
           \begin{aligned}
             & \out {tm}{\Mueval{\Bound z} x} \parallel y \parallel \\
             & \inp{sukm}{y'}{(\out{zerom} \nil \parallel \out{sucm}{\out
                 \km {y'} \parallel \out{initTm} \nil})}
           \end{aligned}
         \right)
       \end{aligned}
               \right)
\end{align*}

The process $\DoneTm$ is executed when we are done computing the term
$t$ and start computing the new testing environment $\envf$. Using the process
representing $\fnm$ saved on the name $\svfn$, it computes the stack
at the bottom of this new environment $\empclos t \cons \fnm+1$, saves
it on the name $\svstk$, and put it also on the name $\newf$ on which
$\envf$ is computed.
\begin{align*}
  \DoneTm & \eqdef \frev(\osef).tm(x).\svfn(y).\left(
            \begin{aligned}
              & \out \newf {\out \hf {\Stk x y} \parallel \out \ef {\tr
                  \mtef}} \parallel \\
              & \out \svstk {\Stk x y} \parallel \out {initRev} \nil
            \end{aligned}
                \right)
  \\
  \Stk {P_t}{P_\fnm} & \eqdef \out \hp {\out {\clos_1}{P_t} \parallel
                       \out {\clos_2}{\tr \mtet} \parallel
                       \out {\clos_3}{\tr \mteb}\parallel
                       \out {\clos_4}{\tr \mtef}} \parallel \out c
                       {\Free{\IndFree {P_\fnm}}}
\end{align*}
The process $\RecRev$ computes $\envf$ on $\newf$ by reversing the environment
that hs been saved on $\rev$. It proceeds by case analysis on the
process on $\rev$, which is either encoding a stack constructor or an
empty stack.
\begin{align*}
  \RecRev & \eqdef initRev(\osef).recrev(x).(x \parallel \out{recrev}
            x \parallel \rev(y).\newf(z). \\
          & \quad \left. \left( 
            \begin{aligned}
              y & \parallel \out{csrev}{\inp \frev \osef {\inp \hrev u
                  {(\out \newf {\out \hf u \parallel \out \ef z}
                    \parallel \out {initRev} \nil)}}} \\
              & \parallel \out {mtrev}{\Done(z)}
            \end{aligned}
                \right) \right)
\end{align*}
During the argument generation, the process $\trf 0$ was stuck
expecting several messages, which can now be provided by the process
$\Done$.
\[
  \Done(P_\envf) = \inp \fdone \osef{\inp \svstk x {(\out \hf z
      \parallel \out \ef \nil \parallel \out \svf {P_\envf} \parallel
      \out \svfn \nil \parallel \out \rev \nil)}}
\]

\end{document}